\documentclass[submission,copyright,creativecommons]{eptcs}
\providecommand{\event}{AiML 2026} 

\usepackage{aiml26}

\usepackage{iftex}

\usepackage[T1]{fontenc}    

\usepackage{latexsym}
\usepackage{amssymb}
\usepackage{amsmath}
\usepackage{stmaryrd}
\usepackage[all, knot]{xy}
\usepackage{verbatim}
\usepackage[english]{babel}
\usepackage{graphicx}
\usepackage{multirow}
\usepackage{wasysym}
\usepackage{tikz}
\usetikzlibrary{arrows.meta,positioning}

\title{A General Theory of Propositional Modal Bundled Modalities}
\author{Yifeng Ding \qquad \qquad Yuanzhe Yang
\institute{Department of Philosophy\\
Peking University\thanks{The authors thank three anonymous referees for their helpful comments, and also the audience at the workshop ``Bundles in Logic'' at Lorentz Center.}
}
\email{\quad yf.ding@pku.edu.cn \quad \qquad 2301210999@pku.edu.cn}
}

\newcommand{\titlerunning}{A General Theory of Propositional Modal Bundled Modalities}
\newcommand{\authorrunning}{Y.Ding, Y.Yang}

\newcommand{\BP}{\mathbf{P}}
\newcommand{\BA}{\mathbf{A}}
\newcommand{\Nat}{\mathbb{N}}
\newcommand{\M}{\mathcal{M}}
\newcommand{\N}{\mathcal{N}}

\newcommand{\LL}{\mathcal{L}}
\newcommand{\TT}{\mathcal{T}}
\newcommand{\bund}{{\textnormal{\Circle}}}
\newcommand{\nbund}{{\textnormal{\CIRCLE}}}
\newcommand{\Nbh}{\mathtt{Nbh}}
\newcommand{\Dom}{\mathtt{Dom}}

\newcommand{\op}{{\mathtt{op}}}
\newcommand{\id}{\iota}
\newcommand{\cp}{{\mathtt{com}}}
\newcommand{\Z}{\mathcal{Z}}
\newcommand{\BL}{\mathbf{L}}

\newcommand{\PC}{\mathtt{Core}^+}
\newcommand{\NC}{\mathtt{Core}^-}

\newcommand{\ext}[1]{\Vert #1 \Vert}

\newcommand{\bis}{\mathrel{\mathchoice%
{\raisebox{.3ex}{$\,
  \underline{\makebox[.7em]{$\leftrightarrow$}}\,$}}%
{\raisebox{.3ex}{$\,
  \underline{\makebox[.7em]{$\leftrightarrow$}}\,$}}%
{\raisebox{.2ex}{$\,
  \underline{\makebox[.5em]{\scriptsize$\leftrightarrow$}}\,$}}%
{\raisebox{.2ex}{$\,
  \underline{\makebox[.5em]{\scriptsize$\leftrightarrow$}}\,$}}}}

\hypersetup{
  bookmarksnumbered,
  pdftitle    = {\titlerunning},
  pdfauthor   = {\authorrunning},
  pdfsubject  = {Modal Logic},               % Consider adding a more appropriate subject or description
  pdfkeywords = {Bundled Modality, Bisimulation, Axiomatization} % Uncomment and enter keywords specific to your paper
}

\begin{document}
\maketitle

\begin{abstract}
In studies of bundled modalities, we encode a complex conceptual notion into the semantics of a single modal operator and study its logic.
Although there is already a substantial body of work on various concrete bundled operators,
we still lack a general understanding of them.
In this paper, we provide a general theory of the expressivity and axiomatization of bundled modalities.
We offer a uniform way to define bisimulations for arbitrary bundled modalities and
justify our definition by the corresponding Hennessy-Milner property.
We also define a special class of bundled modalities called positive-negative-independent bundles.
This class of bundles, together with their duals, cover most bundled modalities studied in the literature,
and their axiomatizations can be done with the help of a more abstract notion of convex neighborhood semantics and corresponding representation results.
As case studies, we axiomatize the ``someone knows'' bundle $\bigvee_{a \in A} \Box_a \phi$ over $S5$-models,
the ``disagreement within group'' bundle $\bigvee_{a, b \in A} \Box_a \phi \wedge \Box_b \neg \phi$ over $KD45$-models,
and the ``belief without knowledge'' bundle $B \phi \wedge \neg K \phi$ over $S4.2$-models.
\end{abstract}

\section{Introduction}\label{sec.introduction}
The idea of \emph{a bundled modality} is to encode a complex conceptual notion into the semantics of a single modal operator
and study the logic that takes this operator as the primitive modality.
The logics of bundled modalities,
which have more complicated semantics yet weaker expressive power compared to full modal logics,
often have interesting behaviors on both conceptual and technical levels.

Research on ``bundled modalities'' started long before the name first appeared in the literature.
For example, the \emph{non-contingency modality} $\triangle \phi := \Box \phi \vee \Box \neg \phi$ received attention since 1966 \cite{HM-1966-CNCML, IH-1995-LNC, SK-1995-MLNC, EZ-1999-CDLN},
and the bundle $\circ \phi := \phi \to \Box \phi$ was studied under different names since 2005 \cite{JM-2005-LEA, CS-2008-LIB, CS-2008-LEA}.
The terminology ``bundle'' was first introduced in recent studies of first-order modal logic \cite{YW-2017-NFEL, ML-2022-BFOML, ML-2023-BFOML, XW-2023-CEB};
during the same period,
not only did studies of the aforementioned modalities revive \cite{JF-2014-AN, JF-2015-CKWH, JF-2020-FCL, JF-2021-UTFBN, DG-2016-RIML, DG-2017-NLUTFB, DG-2021-RIML},
but there also emerged studies of many other bundled modalities, e.g.\ 
false belief $\Box \phi \wedge \neg \phi$ \cite{CS-2011-LFB, DG-2017-NLUTFB, JF-2021-UTFBN, TW-2022-IFB, JF-2025-LFBRB},
strong disagreement between two agents $(\Box_a \phi \wedge \Box_b \neg \phi) \vee (\Box_a \neg \phi \wedge \Box_b \phi)$ \cite{TP-2018-SBD},
someone knows $\bigvee_{a \in A} \Box_a \phi$ \cite{AT-2021-SK, YD-2023-SKWA}, and
secret knowledge $\Box_a \phi \wedge \bigwedge_{b \in A \setminus \{a\}} \Box_a \neg \Box_b \phi$ \cite{ZX-2022-SRI},
just to name a few (more examples can be found in Example \ref{ex.cbund}).

Given the multitudinous papers on concrete bundled modalities,
one naturally wants to find a general theory subsuming them.
There are indeed a few attempts towards such a theory in the literature,
but none of them is completely satisfactory.
For example, \cite{AG-2022-vBAML} introduced a method to automatically define bisimulations for bundled modalities (which are called \emph{molecular connectives} there),
but such a method applies only to a restricted class of modalities called \emph{uniform connectives};
\cite{JF-2025-UAN} introduced a uniform approach to axiomatizing logics of bundled modalities called \emph{the method based on almost-definability schemas},
but no uniform way to find such schemas for an arbitrary bundled modality is provided,
and it is also unclear to which kinds of bundled modalities this method applies.
As a result, in studies of concrete bundled modalities,
bisimulations are often left as open problems,
while radically different notions of canonical models are introduced to prove completeness theorems,
even though the logics in question often have a lot in common.

Our goal in this paper, then, is to provide a general theory for both expressivity and axiomatization of bundled modalities.
For expressivity, 
we view the semantics of bundled modalities as a special kind of \emph{neighborhood semantics}
and uniformly define bisimulations for \emph{all} bundled modalities.
As for axiomatizations, we focus on a particularly well-behaved class of bundled modalities called \emph{positive-negative-independent bundles} (\emph{PN-independent bundles} for short).
This class of bundles satisfy a natural convexity axiom and,
together with their duals,
cover most bundled modalities studied in the literature and yet a lot more.
We show that the fairly standard two-step strategy of proving completeness theorems, namely canonical model building followed by transforming to a semantically equivalent intended model, applies with a high degree of uniformity.
Specifically, canonical models can be constructed uniformly with a rich supply of canonical axioms, and model transformation does not require much action since, from our analysis of bundled semantics, we know exactly what it takes to preserve truth values.

The structure of the paper is as follows.
In Section \ref{sec.language}, we introduce the general language and semantics for bundled modalities.
Then, in Section \ref{sec.bisimulation}, we show how to uniformly define bisimulations for bundled modalities
and prove a Hennessy-Milner theorem for the resulting bisimulations.
In Section \ref{sec.convex}, we introduce the notion of PN-independent bundles,
characterize their expressivity,
and also introduce a novel notion of \emph{convex neighborhood logic},\footnote{
  The same kind of logic is also discussed in \cite{AT-2021-SK}.}
which lies between minimal neighborhood logic and monotonic neighborhood logic,
and abstracts the core feature shared by all PN-independent bundles.
Then, in Section \ref{sec.cases}, to illustrate our recipe for axiomatizations, 
we offer complete axiomatizations for three concrete bundled modalities:
the ``someone knows'' bundle $\bigvee_{a \in A} \Box_a \phi$ \cite{AT-2021-SK,YD-2023-SKWA} over $S5$-models,
the ``disagreement within group'' bundle $\bigvee_{a, b \in A} \Box_a \phi \wedge \Box_b \neg \phi$ (which generalizes the two-agent case \cite{TP-2018-SBD}) over $KD45$-models,
and the ``belief without knowledge'' bundle $B \phi \wedge \neg K \phi$ over $S4.2$-models (and also $S4F$-models).
To our knowledge, none of them is available in the literature.
Finally, in Section \ref{sec.conclusion},
we summarize our results and suggest some directions for future work.

\section{Language and Semantics}\label{sec.language}
We now define the language and semantics for bundled modalities in general.
For simplicity, we focus on cases where there is only one unary bundled modality;
generalization to multiple bundled modalities with arbitrary arities is left for future work.
Formally, we introduce only \emph{one} modal language with a unary modal operator,
and different bundled modalities are viewed as different \emph{semantics} for the same modal operator.

Fix a countable set $\BP$ of propositional letters.

\begin{definition}[Language]
  Formulas in $\LL_\bund$ are defined recursively as follows:
  \begin{center}
    $\phi ::= p \mid \top \mid \neg \phi \mid (\phi \wedge \phi) \mid \bund \phi \quad (\text{where } p \in \BP)$.
  \end{center}

  Let $\nbund \phi$ be the abbreviation for $\bund \neg \phi$ ($\nbund \phi$ becomes useful when we offer axiomatizations).
\end{definition}

Next, we define the \emph{bundled semantics} for the above language.
We fix a countable set $\BA$ of indices,
together with an arity function $\rho$ from $\BA$ to $\Nat$
s.t.\ for any $n \in \Nat$, $\rho^{-1}[\{n\}]$ is infinite.
We also assume that for each $n \in \Nat$,
there is $\id_n \in \BA$ with $\rho(\id_n) = n$.
We consider Kripke models of the following kind.

\begin{definition}[Kripke Models]
  A \emph{Kripke model} is a triple $\M = (W, \{R_a\}_{a \in \BA}, V)$, where
  $W \neq \emptyset$ is the set of worlds;
  for each $a \in \BA$, $R_a$ is a $(\rho(a) + 1)$-ary relation on $W$
  and in particular, $R_{\id_n} = \bigcup_{w \in W} \{w\}^{n+1}$;
  and $V$ is a valuation function from $\BP$ to $\wp(W)$.
  For $w \in W$ and $a \in \BA$, we use $R_a(w)$ to denote $\{(v_0, ..., v_{n-1}) \in W^n \mid (w, v_0, ..., v_{n-1}) \in R_a\}$ and do not distinguish $(v)$ from $v$.
\end{definition}

The structures of the bundles are expressed by the following \emph{bundle terms}, which are essentially formulas of a polyadic modal language with only one propositional letter in \emph{negation normal form} (where $(+)$ is the only propositional letter, and $(-)$ is its negation):

\begin{definition}[Bundle Terms]
  Terms in $\TT_\mathbf{A}$ are defined recursively as follows:
  \begin{center}
    $\tau ::= (+) \mid (-) \mid \nabla_a(\underbrace{\tau, ..., \tau}_{\rho(a)\text{-many}}) \mid \Delta_a(\underbrace{\tau, ..., \tau}_{\rho(a)\text{-many}}) \quad (\text{where } a \in \BA)$.
  \end{center}
  
  As usual, when $\rho(a) = 1$, we may use $\Box_a \tau$ and $\Diamond_a \tau$ to denote $\nabla_a(\tau)$ and $\Delta_a(\tau)$, respectively.
\end{definition}

Note that Boolean connectives are not included as primitive components in our definition of bundle terms:
this is because they can be viewed as \emph{modalities} corresponding to the $\id_n$'s.
For example, binary disjunction and conjunction can be viewed as $\nabla_{\id_2}$ and $\Delta_{\id_2}$, respectively.\footnote{
The idea of viewing Boolean operators as modalities can also be found in \cite{VG-1998-SFPM, VG-2006-INDF}, for example.
}

Then, for any bundle term $\tau$, we define the corresponding bundled semantics for $\LL_\bund$ as follows.

\begin{definition}[Bundled Semantics]
  Let $\M = (W, \{R_a\}_{a \in \BA}, V)$ be a Kripke model.
  
  First, for any $w \in W$ and $X \subseteq W$, we define the following satisfaction relation $\vDash$ for terms in $\TT_\BA$:
  \begin{center}
    \begin{tabular}{|lcl|}
      \hline
      $\M,w,X \vDash (+)$ & $\iff$ & $w \in X$\\
      $\M,w,X \vDash (-)$ & $\iff$ & $w \notin X$\\
      $\M,w,X \vDash \nabla_a(\tau_0, ..., \tau_{n-1})$ & $\iff$ & $\forall (v_0, ..., v_{n-1}) \in R_a(w) \exists i < n \M,v_i,X \vDash \tau_i$\\
      $\M,w,X \vDash \Delta_a(\tau_0, ..., \tau_{n-1})$ & $\iff$ & $\exists (v_0, ..., v_{n-1}) \in R_a(w) \forall i < n \M,v_i,X \vDash \tau_i$.\\
      \hline
    \end{tabular}
  \end{center}
  
  Then, for any $\tau \in \TT_\BA$, $w \in W$ and $\phi \in \LL_\bund$, 
  the satisfaction relation $\vDash_\tau$ for $\tau$ is defined recursively as follows
  (Boolean cases are omitted, and $|\phi|^\tau_\M = \{w \in W \mid \M,w \vDash_\tau \phi\}$):
  \begin{center}
    \begin{tabular}{|lcl|}
      \hline
      $\M,w \vDash_\tau p$ & $\iff$ & $w \in V(p)$\\
      $\M,w \vDash_\tau \bund \phi$ & $\iff$ & $\M,w,|\phi|^\tau_\M \vDash \tau$.\\
      \hline
    \end{tabular}
  \end{center}
\end{definition}

Intuitively, we ``unzip'' the modal structure encoded by $\tau$ when we compute the truth value of $\bund \phi$.
This formalizes the usual way in which semantics of bundled modalities are defined.

Note that since $R_{\id_2} = \{(w, w, w) \mid w \in W\}$,
$\M,w,X \vDash \nabla_{\id_2}(\tau_0,\tau_1)$ iff $(\M,w,X \vDash \tau_0) \vee (\M,w,X \vDash \tau_1)$,
while $\M,w,X \vDash \Delta_{\id_2}(\tau_0,\tau_1)$ iff $(\M,w,X \vDash \tau_0) \wedge (\M,w,X \vDash \tau_1)$,
so $\nabla_{\id_2}$ and $\Delta_{\id_2}$ can be safely viewed as $\vee$ and $\wedge$.
For the sake of readability, we shall directly write $\vee$ and $\wedge$ rather than $\nabla_{\id_2}$ and $\Delta_{\id_2}$ from now on.

As a concrete example, if we want to give $\bund \phi$ the non-contingency semantics,
we may take some $a \in \BA$ s.t.\ $\rho(a) = 1$,
and let $\tau_\triangle = \Box_a (+) \vee \Box_a (-)$.
Then, it is not hard to see that
$\M,w \vDash_{\tau_\triangle} \bund \phi$ iff $(\forall v \in R_a(w) \M,v \vDash_{\tau_\triangle} \phi) \vee (\forall v \in R_a(w) \M,v \not \vDash_{\tau_\triangle} \phi)$, as intended.

\section{Bisimulations for Bundled Modalities}\label{sec.bisimulation}
In this section, we offer a uniform way to define bisimulations for bundled modalities
and show that the bisimulations defined in this way correctly characterize the expressivity of the corresponding bundled modalities.
A key observation leading to these bisimulations is that
the bundled semantics can be viewed as a special kind of \emph{neighborhood semantics}.
The neighborhood structure generated by a bundle term over a Kripke model is defined as follows 
(from which one could see why we choose to define bundle terms in negation normal forms).

\begin{definition}[Generated Neighborhood]
  Let $\M = (W, \{R_a\}_{a \in \BA}, V)$ be a Kripke model. 
  
  We define a function $\Nbh_\M$ from $W \times \TT_\BA$ to $\wp(\wp(W) \times \wp(W))$ recursively as follows:\footnote{
    $\prod_{x \in X} Y_x$ denotes the set of functions from $X$ to $\bigcup_{x \in X} Y_x$
    that map each $x \in X$ to an element in $Y_x$.
  }
  \begin{center}
    \begin{tabular}{rcl}
    $\Nbh_\M(w, (+))$ & $=$ & $\{(\{w\}, \emptyset)\}$\\
    $\Nbh_\M(w, (-))$ & $=$ & $\{(\emptyset, \{w\})\}$\\
    $\Nbh_\M(w, \nabla_a(\tau_0, ..., \tau_{n-1}))$ & $=$ & $\{\bigsqcup \text{ran}(f) \mid f \in \prod_{(v_0, ..., v_{n-1}) \in R_a(w)} \bigcup_{i < n} \Nbh_\M(v_i, \tau_i)\}$\\
    $\Nbh_\M(w, \Delta_a(\tau_0, ..., \tau_{n-1}))$ & $=$ & $\bigcup_{(v_0, ..., v_{n-1}) \in R_a(w)} \{\bigsqcup \text{ran}(g) \mid g \in \prod_{i < n} \Nbh_\M(v_i, \tau_i)\}$\\[0.3em]
    && (where $\bigsqcup\{ (X_i, Y_i) \mid {i \in I}\}= (\bigcup_{i \in I} X_i, \bigcup_{i \in I} Y_i)$)
    \end{tabular}
  \end{center}
\end{definition}

Intuitively, given a pointed model $\M,w$,
each $(X, Y) \in \Nbh_\M(w, \tau)$ represents a minimal way to satisfy and falsify $\phi$ in the model 
in order for $\bund \phi$ to be true at $\M,w$ under the $\tau$-bundled semantics
($X$ and $Y$ mark where $\phi$ must be satisfied and falsified, respectively).
Following this intuition,
an element in $\Nbh_\M(w, \nabla_a(\tau_0, ..., \tau_{n-1}))$ is obtained by choosing for each successor $(v_0, ..., v_{n-1})$ of $w$ one $i < n$ and one element from $\Nbh_\M(v_i, \tau_i)$, and then merging them together,
while an element in $\Nbh_\M(w, \Delta_a(\tau_0, ..., \tau_{n-1}))$ is obtained by first choosing one successor $(v_0, ..., v_{n-1})$ of $w$, then choosing for each $i < n$ one element from $\Nbh_\M(v_i, \tau_i)$, and merging them together.

Our intuition above is verified by the following proposition,
which shows that $\bund \phi$ is true at $\M,w$ iff there is a pair $(X, Y) \in \Nbh_\M(w, \tau)$
s.t.\ $X \subseteq |\phi|_\M$ and $Y \subseteq W \setminus |\phi|_\M^\tau$
--- in other words, $X$ and $W \setminus Y$ constitute a pair of lower and upper bounds for $|\phi|^\tau_\M$.
It is worth noting that the \emph{Axiom of Choice} (AC) is indispensable for its proof.

For simplicity, we write $(X, Y) \sqsubseteq (X', Y')$ for $X \subseteq X'$ and $Y \subseteq Y'$.

\begin{proposition}\label{prop.nbhsem1}(AC)
  For any model $\M,w$, any $\tau \in \TT_\BA$ and $\phi \in \LL_\bund$, 
  \begin{center}
    $\M,w \vDash_\tau \bund \phi \iff \exists (X, Y) \in \Nbh_\M(w, \tau) ((X, Y) \sqsubseteq (|\phi|^\tau_\M, W \setminus |\phi|^\tau_\M))$.
  \end{center}
\end{proposition}

\begin{proof}
  Fix a model $\M = (W, \{R_a\}_{a \in \BA}, V)$.
  We use induction on the structure of terms to show that for any $\tau \in \TT_\BA$, $w \in W$ and $Z \subseteq W$,
  $\M,w,Z \vDash \tau$ iff $\exists (X, Y) \in \Nbh_\M(w, \tau) ((X, Y) \sqsubseteq (Z, W \setminus Z))$.

  The base cases for $(+)$ and $(-)$ are obvious by definition, so we focus on the inductive steps.

  For $\nabla_a(\tau_0, ..., \tau_{n-1})$, we have 
  \begin{align*}
    \M,w,Z \vDash \nabla_a(\tau_0, ..., \tau_{n-1}) \iff & \forall \vec{v} \in R_a(w) \exists i < n \M,v_i,Z \vDash \tau_i\\
    \overset{\text{IH}}{\iff} & \forall \vec{v} \in R_a(w) \exists i < n \exists (X, Y) \in \Nbh_\M(v_i, \tau_i) ((X, Y) \sqsubseteq (Z, W \setminus Z)) \\
    \overset{\text{AC}}{\iff} & \exists f \in \prod_{\vec{v} \in R_a(w)} \bigcup_{i < n} \Nbh_\M(v_i, \tau_i) \forall \vec{v} \in R_a(w) (f(\vec{v}) \sqsubseteq (Z, W \setminus Z)) \\
    \iff & \exists f \in \prod_{\vec{v} \in R_a(w)} \bigcup_{i < n} \Nbh_\M(v_i, \tau_i) (\bigsqcup \text{ran}(f) \sqsubseteq (Z, W \setminus Z))\\
    \iff & \exists (X, Y) \in \Nbh_\M(w, \nabla_a(\tau_0, ..., \tau_{n-1})) ((X, Y) \sqsubseteq (Z, W \setminus Z)).
  \end{align*}

  For $\Delta_a(\tau_0, ..., \tau_{n-1})$, similarly, we have 
  \begin{align*}
    \M,w,Z \vDash \Delta_a(\tau_0, ..., \tau_{n-1}) \iff & \exists \vec{v} \in R_a(w) \forall i < n \M,v_i,Z \vDash \tau_i\\
    \overset{\text{IH}}{\iff} & \exists \vec{v} \in R_a(w) \forall i < n \exists (X, Y) \in \Nbh_\M(v_i, \tau_i) ((X, Y) \sqsubseteq (Z, W \setminus Z))\\
    \iff & \exists \vec{v} \in R_a(w) \exists g \in \prod_{i < n} \Nbh_\M(v_i, \tau_i) \forall i < n (g(i) \sqsubseteq (Z, W \setminus Z))\\
    \iff & \exists \vec{v} \in R_a(w) \exists g \in \prod_{i < n} \Nbh_\M(v_i, \tau_i) (\bigsqcup \text{ran}(g) \sqsubseteq (Z, W \setminus Z))\\
    \iff & \exists (X, Y) \in \Nbh_\M(w, \Delta_a(\tau_0, ..., \tau_{n-1})) ((X, Y) \sqsubseteq (Z, W \setminus Z)).
  \end{align*}

  Note that we do not need AC in the case for $\Delta_a(\tau_0, ..., \tau_{n-1})$ since the $g$ we need is finite.
\end{proof}

Before moving on to define the bisimulations,
as preparation,
we first introduce the following notions of \emph{domains} and \emph{completions} for the neighborhood structures generated by terms.

\begin{definition}[Domain and Completion]
  Let $\M = (W, \{R_a\}_{a \in \BA}, V)$ be a model.
  
  We define a function $\Dom_\M$ from $W \times \TT_\BA$ to $\wp(W)$ recursively as follows:
  \begin{center}
    \begin{tabular}{rcl}
      $\Dom_\M(w, (+))$, $\Dom_\M(w, (-))$ & $=$ & $\{w\}$ \\
      $\Dom_\M(w, \nabla_a(\tau_0, ..., \tau_{n-1}))$, $\Dom_\M(w, \Delta_a(\tau_0, ..., \tau_{n-1}))$ & $=$ & $\bigcup_{(v_0, ..., v_{n-1}) \in R_a(w)} \bigcup_{i < n} \Dom_\M(v_i, \tau_i)$.\\
    \end{tabular}
  \end{center}

  Then, for any term $\tau \in \TT_\BA$, the \emph{completion} of $\Nbh_\M(w, \tau)$
  is defined as 
  \begin{center}
    $\Nbh^\cp_\M(w, \tau) = \{Z \subseteq \Dom_\M(w, \tau) \mid \exists (X, Y) \in \Nbh_\M(w, \tau) ((X, Y) \sqsubseteq (Z, \Dom_\M(w, \tau) \setminus Z))\}$.
  \end{center}
\end{definition}

Intuitively, the truth value of $\bund \phi$ at $\M,w$ under $\tau$-semantics depends only on the truth values of $\phi$ in $\mathtt{Dom}_\M(w, \tau)$.
This is formally expressed by following proposition, 
which can be checked via an easy induction on the structure of terms.

\begin{proposition}\label{prop.nbhsem2}
  For any model $\M,w$ and any $\tau \in \TT_\BA$, $\bigcup_{(X, Y) \in \Nbh_\M(w, \tau)} X \cup Y \subseteq \Dom_\M(w, \tau)$.
  Thus, by Proposition \ref{prop.nbhsem1}, for any $\phi \in \LL_\bund$, $\M,w \vDash_\tau \bund \phi \iff |\phi|^\tau_\M \cap \Dom_\M(w, \tau) \in \Nbh^\cp_\M(w, \tau)$.
\end{proposition}

With the help of the above notions, we can define bisimulations for bundled modalities in a uniform way.
Essentially, we are just applying the kind of neighborhood bisimulation introduced in \cite{HH-2007-NBHBIS} to the neighborhood structures generated by bundle terms,
but there are also some minor adjustments.

\begin{definition}[Bundled Bisimulation]
  Let $\M = (W, \{R_a\}_{a \in \BA}, V)$ be a model, and let $\tau \in \TT_\BA$.
  Let $\Z$ be a binary relation on $W$.
  For any $w, v \in W$ and $X \subseteq \mathtt{Dom}_\M(w, \tau)$, $Y \subseteq \mathtt{Dom}_\M(v, \tau)$,
  we say $X$ and $Y$ are \emph{$\Z$-coherent},
  if $\Z[X \cup Y] \cap \mathtt{Dom}_\M(w, \tau) \subseteq X$
  and $\Z[X \cup Y] \cap \mathtt{Dom}_\M(v, \tau) \subseteq Y$.

  Then, we say $\Z$ is a \emph{$\tau$-bisimulation} on $\M$, if for all $(w, v) \in \Z$,
  the following are satisfied:
  \begin{itemize}
    \item (Inv) For all $p \in \mathbf{P}$, $w \in V(p)$ iff $v \in V(p)$;
    \item ($\tau$-Coh) For all $X \subseteq \Dom_\M(w, \tau)$ and $Y \subseteq \Dom_\M(v, \tau)$ s.t.\ $X$ and $Y$ are $\Z$-coherent,\\
    $X \in \Nbh^\cp_\M(w, \tau)$ iff $Y \in \Nbh^\cp_\M(v, \tau)$.
  \end{itemize}

  We say $\M,w$ and $\N,v$ are \emph{$\tau$-bisimilar},
  denoted by $\M,w \bis_\tau \N,v$,
  if there is a $\tau$-bisimulation $\Z$ on the disjoint union of $\M$ and $\N$ s.t.\ $(w, v) \in \Z$.
\end{definition}

Note that the condition ($\tau$-Coh) only quantifies over pairs of \emph{$\Z$-coherent} sets,
from which many (but maybe not all) pairs of sets that are undefinable over $\Dom_\M(w, \tau)$ and $\Dom_\M(v, \tau)$ are excluded.
If we quantify over all pairs of sets instead,
than the resulting notion of bisimulation would become too strong,
and the Hennessy-Milner theorem (Theorem \ref{thm.hm}) would no longer hold.
One could refer to \cite{HH-2007-NBHBIS} for more details concerning the notion of $\Z$-coherence
in the more abstract context of neighborhood logic.

It is rather routine to check that $\bis_\tau$ is indeed an equivalence relation.
It is also interesting to note that $\tau$-bisimulation is \emph{not} first-order definable when $\tau$ is complicated enough; see Appendix \ref{appd.fodf} for an example.
Below is a simple example of a pair of $\tau$-bisimilar models.

\begin{example}\label{ex.bis}
  Let $\tau = \Diamond \Box \Diamond (+) \wedge \Diamond \Box \Diamond (-)$, and let $\M$, $\N$ be the following models.
  
  $\M,w_0$ and $\N,v_0$ are not bisimilar in the standard sense, but they are $\tau$-bisimilar.
  \begin{center}
  {\small
    $
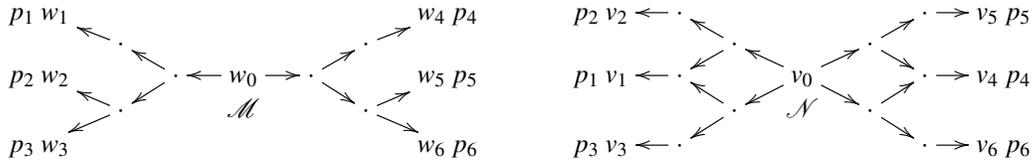
\xymatrix@R=0.05em@C=1.2em{
    p_1 \; w_1 & & & & & & w_4 \; p_4 & &  p_2 \; v_2 & \cdot \ar[l] & & & & \cdot \ar[r] & v_5 \; p_5 \\
     & \cdot \ar[ul] & & & & \cdot \ar[ur] & & & & & \cdot \ar[dl] \ar[ul] & & \cdot \ar[dr] \ar[ur]  \\
    p_2 \; w_2 & & \cdot \ar[ul] \ar[dl] & w_0 \ar[l] \ar[r] & \cdot \ar[ur] \ar[dr] & & w_5 \; p_5 & & p_1 \; v_1 & \cdot \ar[l] & & v_0 \ar[ul] \ar[dl] \ar[dr] \ar[ur] & & \cdot \ar[r] & v_4 \; p_4\\
     & \cdot \ar[ul] \ar[dl] & & \M & & \cdot \ar[ur] \ar[dr] & & & & & \cdot \ar[ul] \ar[dl] & \N & \cdot \ar[ur] \ar[dr]\\
    p_3 \; w_3 & & & & & & w_6 \; p_6 & & p_3 \; v_3 & \cdot \ar[l] & & & & \cdot \ar[r] & v_6 \; p_6\\
    }$\\
  }
  \end{center}
  \begin{itemize}
    \item $\Dom_\M(w_0, \tau) = \{w_i \mid 1 \leq i \leq 6\}$ and $\Dom_\N(v_0, \tau) = \{v_i \mid 1 \leq i \leq 6\}$.
    \item $\Nbh_\M(w_0, \tau) = N \times N$, where $N = \{\{w_1, w_2\}, \{w_1, w_3\}, \{w_4, w_5\}, \{w_4, w_6\}\}$.
    \item $\Nbh_\N(v_0, \tau) = N' \times N'$, where $N' = \{\{v_1, v_2\}, \{v_1, v_3\}, \{v_4, v_5\}, \{v_4, v_6\}\}$.
    \item For any $1 \leq i \leq 6$, $\Dom_\M(w_i, \tau) = \Dom_\N(v_i, \tau) = \Nbh_\M(w_i, \tau) = \Nbh_\N(v_i, \tau) = \emptyset$.
    \item $\Z = \bigcup_{0 \leq i \leq 6} (\{w_i, v_i\} \times \{w_i, v_i\})$
    forms a $\tau$-bisimulation on the disjoint union of $\M$ and $\N$,
    so $\M, w_0 \bis_\tau \N, v_0$.
  \end{itemize}
\end{example}

\begin{comment} % Standard bisimulation implies bundled bisimulation
It is also worth noting that the above notion of bisimilarity is implied by the standard notion of bisimilarity $\bis$.\footnote{
We do not include the definition here due to limited space.
The reader can refer to textbooks like \cite{PB-2001-ML} for the standard definition of bisimilarity.
}

\begin{proposition}\label{prop.stdbis}(AC)
  For any models $\M,w$, $\N,v$ and $\tau \in \TT_\BA$, $\M,w \bis \N,v \implies \M,w \bis_\tau \N,v$.
\end{proposition}

\begin{proof}[Proof Sketch]
  It suffices to show that for any model $\M$, if $\Z$ is a standard bisimuluation on $\M$, then it is also $\tau$-bisimulation on $\M$.
  Assume that $\Z$ is a standard bisimulation on $\M$.
  Without loss of generality, we may also assume that $\Z$ is an equivalence relation.
  Then, by induction on the structure of terms,
  we can show that for any $\tau \in \TT_\BA$, $(w, v) \in \Z$ and
  $(X_0, X_1) \in \Nbh_\M(w, \tau)$, there is $(Y_0, Y_1) \in \Nbh_\M(v, \tau)$ s.t.\ 
  for all $i \in \{0, 1\}$ and $v' \in Y_i$, there is $w' \in X_i$ s.t.\ $(w', v') \in \Z$.
  The rest of the proof is similar with the proof for the first part of Proposition \ref{prop.cvbis}.
\end{proof}
\end{comment}

Then, we check that the truth values of formulas under bundled semantics are indeed preserved by bundled bisimilarity.
For simplicity, we write $\M,w \equiv_\tau \N,v$ when 
$\M,w \vDash_\tau \phi \iff \N,v \vDash_\tau \phi$ holds for all $\phi \in \LL_\bund$.

\begin{proposition}
  For any models $\M,w$ and $\N,v$ and $\tau \in \TT_\BA$,
  $\M,w \bis_\tau \N,v \implies \M,w \equiv_\tau \N,v$.
\end{proposition}

\begin{proof}
  It suffices to show by induction that for any model $\M$ and $\tau$-bisimulation $\Z$ on $\M$,
  for any $\phi \in \LL_\bund$,
  if $(w, v) \in \Z$, then $\M,w \vDash_\tau \phi$ iff $\M,v \vDash_\tau \phi$.
  We only show the case for $\bund \phi$.

  Assume that $(w, v) \in \Z$,
  and let $X = |\phi|^\tau_\M \cap \Dom_\M(w, \tau)$, $Y = |\phi|^\tau_\M \cap \Dom_\M(v, \tau)$.
  By IH, it is not hard to see that $X$ and $Y$ are $\Z$-coherent,
  so $X \in \Nbh^\cp_\M(w, \tau)$ iff $Y \in \Nbh^\cp_\M(v, \tau)$.
  Then, by proposition \ref{prop.nbhsem2}, $\M,w \vDash_\tau \bund \phi$ iff $X \in \Nbh^\cp_\M(w, \tau)$ iff $Y \in \Nbh^\cp_\M(v, \tau)$ iff $\M, v \vDash_\tau \bund \phi$.
\end{proof}

Finally, we show that the converse direction of the above implication holds when the models in question are \emph{image-finite} (i.e.\ for any $w \in W$ and $a \in \BA$,
$R_a(w)$ is finite).
This shows that the notion of $\tau$-bisimilarity indeed characterizes the expressivity of $\LL_\bund$ under $\tau$-semantics in a natural way.

\begin{theorem}[Hennessy-Milner]\label{thm.hm}
  For any image-finite models $\M,w$, $\N,v$ and $\tau \in \TT_\BA$,
  \begin{center}
    $\M,w \bis_\tau \N,v \iff \M,w \equiv_\tau \N,v$.
  \end{center}
\end{theorem}

\begin{proof}
  Let $\tau \in \TT_\BA$ be arbitrary.
  Using disjoint unions, it suffices to show that $\Z = \{(w, v) \mid \M,w \equiv_\tau \M,v\}$ forms a $\tau$-bisimulation for any image-finite $\M$.

  The condition (Inv) is trivially satisfied, so we only need to check ($\tau$-Coh).
  Assume that $(w, v) \in \Z$.
  Since $\M$ is image-finite, it is easy to check that $\Dom_\M(w, \tau)$ and $\Dom_\M(v, \tau)$ are finite,
  so we can define every $\Z$-equivalence class $C$ over $\Dom_\M(w, \tau) \cup \Dom_\M(v, \tau)$ with a formula $\psi_C$.
  Now, assume that $X \subseteq \Dom_\M(w, \tau)$ and $Y \subseteq \Dom_\M(v, \tau)$ are $\Z$-coherent,
  and consider $\phi = \bigvee_{u \in X \cup Y} \psi_{[u]_\Z}$,
  where $[u]_\Z$ is the $\Z$-equivalence class over $\Dom_\M(w, \tau) \cup \Dom_\M(v, \tau)$ containing $u$.
  Since $X$ and $Y$ are $\Z$-coherent,
  it is not hard to check that $X = |\phi|^\tau_\M \cap \Dom_\M(w, \tau)$
  and $Y = |\phi|^\tau_\M \cap \Dom_\M(v, \tau)$.
  Then, since $\M,w \equiv_\tau \M,v$,
  $\M,w \vDash_\tau \bund \phi$ iff $\M,v \vDash_\tau \bund \phi$,
  so by Proposition \ref{prop.nbhsem2}, $X \in \Nbh^\cp_\M(w, \tau)$ iff $Y \in \Nbh^\cp_\M(v, \tau)$.
\end{proof}

\section{PN-Independent Bundled Modalities}\label{sec.convex}

Now, we turn to a special class of bundled modalities we call \emph{positive-negative-independent} bundled modalities (\emph{PN-independent} bundles for short).
As we shall see, most bundled modalities studied in the literature are either PN-independent or duals of PN-independent bundles, 
and such bundles are indeed well-behaved on the technical level.

\begin{definition}[PN-Independent Bundles]\label{def.cvb}
  Let $\M = (W, \{R_a\}_{a \in \BA}, V)$ be a model, and let $\tau \in \TT_\BA$.
  We say $\tau$ is \emph{PN-independent} over $\M$,
  if for all $w \in W$ and $(X, Y), (X', Y') \in \Nbh_\M(w, \tau)$,
  if $X \cap Y = X' \cap Y' = X \cap Y' = \emptyset$,
  then $(X, Y') \in \Nbh_\M(w, \tau)$.
\end{definition}

We say a family $N$ of sets is \emph{convex},
if for any $X, Y \in N$ and any set $Z$,
if $X \subseteq Z \subseteq Y$,
then $Z \in N$.
Note that if $\tau$ is PN-independent over $\M$,
then for any $w$ in $\M$,
$\Nbh^\cp_\M(w, \tau)$ is convex.
This feature of PN-independent modalities validates the axiom of convexity 
$\mathtt{CONV}: \bund (\phi \wedge \psi) \wedge \bund (\phi \vee \chi) \to \bund \phi$,
which, as we shall see,
plays a crucial role in axiomatizations of PN-independent bundles.

There is an easy syntactic sufficient condition for PN-independence,
which covers a great amount of PN-independent bundles in the literature in their exact syntactic forms,
and probably all such bundles in the literature 
up to logical equivalence.

\begin{proposition}\label{prop.cvnf}
Define bundle terms in $\TT^\times_\BA$ by the following grammar:
\begin{itemize}
  \item[] $\TT^+_\BA \ni \tau^+ ::= (+) \mid \nabla_a(\tau^+, ..., \tau^+) \mid \Delta_a(\tau^+, ..., \tau^+)$ \quad ($a \in \BA$)
  \item[] $\TT^-_\BA \ni \tau^- ::= (-) \mid \nabla_a(\tau^-, ..., \tau^-) \mid \Delta_a(\tau^-, ..., \tau^-)$ \quad ($a \in \BA$)
  \item[] $\TT^\times_\BA \ni \tau^\times ::= \tau^+ \mid \tau^- \mid \tau^\times \wedge \tau^\times \mid \Box_a \tau^\times$ \quad ($\tau^+ \in \TT_\BA^+$, $\tau^- \in \TT_\BA^-$, $\rho(a) = 1$).
\end{itemize}
  Then, for any $\tau^\times \in \TT^\times_\BA$
  and any model $\M$,
  $\tau^\times$ is PN-independent over $\M$.
\end{proposition}

\begin{proof}[Proof Sketch]
  Let $\M = (W, \{R_a\}_{a \in \BA}, V)$ be arbitrary.
  First, note that for any $w \in W$ and $\tau^+ \in \TT^+_\BA$,
  there is $N^+_{w, \tau^+} \in \wp(\wp(W))$ s.t.\ $\Nbh_\M(w, \tau^+) = N^+_{w, \tau^+} \times \{\emptyset\}$;
  similarly, for any $\tau^- \in \TT^-_\BA$,
  there is $N^-_{w, \tau^-} \in \wp(\wp(W))$ s.t.\ $\Nbh_\M(w, \tau^-) = \{\emptyset\} \times N^-_{w, \tau^-}$.
  Then, we check by induction that for any $w \in W$ and $\tau^\times \in \TT^\times_\BA$,
  there are $N^+_{w, \tau^\times}, N^-_{w, \tau^\times} \in \wp(\wp(W))$
  s.t.\ $\Nbh_\M(w, \tau^\times) = N^+_{w, \tau^\times} \times N^-_{w, \tau^\times}$ (which is actually stronger than the requirement in Definition \ref{def.cvb}).
  The base case for $\tau^+$ and $\tau^-$ was just
  noted above.
  For the inductive step,
  if $\Nbh_\M(w, \tau^\times_0) = N^+_{w, \tau^\times_0} \times N^-_{w, \tau^\times_0}$
  and $\Nbh_\M(w, \tau^\times_1) = N^+_{w, \tau^\times_1} \times N^-_{w, \tau^\times_1}$,
  then it is routine to compute that 
  $\Nbh_\M(w, \tau^\times_0 \wedge \tau^\times_1) = \{X_0 \cup X_1 \mid X_0 \in N^+_{w, \tau^\times_0}, X_1 \in N^+_{w, \tau^\times_1}\} \times \{Y_0 \cup Y_1 \mid Y_0 \in N^-_{w, \tau^\times_0}, Y_1 \in N^-_{w, \tau^\times_1}\}$;
  similarly, if $\Nbh_\M(v, \tau^\times) = N^+_{v, \tau^\times} \times N^-_{v, \tau^\times}$ for all $v \in R_a(w)$,
  then $\Nbh_\M(w, \Box_a \tau^\times) = \{\bigcup \text{ran}(g) \mid g \in \prod_{v \in R_a(w)} N^+_{v, \tau^\times}\} \times \{\bigcup \text{ran}(h) \mid h \in \prod_{v \in R_a(w)} N^-_{v, \tau^\times}\}$.
\end{proof}

\begin{example}\label{ex.cbund}
  By Proposition \ref{prop.cvnf}, the following bundles are PN-independent over all models:
  $\Diamond (+) \wedge \Diamond (-)$ (contingency / ignorance \cite{HM-1966-CNCML, IH-1995-LNC, SK-1995-MLNC, EZ-1999-CDLN, JF-2014-AN, JF-2015-CKWH, JF-2020-FCL});
  $(+) \wedge \Diamond (-)$ (unknown truth / accident \cite{JM-2005-LEA, CS-2008-LIB, CS-2008-LEA, JF-2021-UTFBN, DG-2016-RIML, DG-2017-NLUTFB, DG-2021-RIML});
  $\Box (+) \wedge (-)$ (false belief \cite{CS-2011-LFB,DG-2017-NLUTFB, JF-2021-UTFBN, TW-2022-IFB, JF-2025-LFBRB});
  $\Box (-) \wedge (+)$ (strong accident \cite{TP-2017-WESA});
  $((+) \vee \Diamond (+)) \wedge \Diamond (-)$ (semi-weak contingency / disjunctive ignorance \cite{JF-2021-DISI});
  $((+) \vee \Diamond (+)) \wedge ((-) \vee \Diamond (-))$ (weak contingency \cite{JF-2019-SNC});
  $BK (+) \wedge \widehat{K} (-)$ (Dunning-Kruger ignorance \cite{JF-2025-RIDKI});
  $\Box_a (+) \wedge \Diamond_b (-)$, $\Box_a (+) \wedge \Box_b (-)$, $\Diamond_a (+) \wedge \Box_b (-)$, $\Diamond_a (+) \wedge \Diamond_b (-)$ (\cite{JF-2025-UAN});
  $\bigvee_{a \in A} \Box_a (+)$ (someone knows \cite{AT-2021-SK,YD-2023-SKWA});
  $\Box_a (+) \wedge \bigwedge_{b \in B} \Box_a \Diamond_b (-)$ (secret knowledge \cite{ZX-2022-SRI});
  $(\Box_a(+) \wedge \Diamond_b(-)) \vee (\Box_b(+) \wedge \Diamond_a(-))$ (moderate disagreement \cite{TP-2019-MBD}; it is equivalent to $(\Box_a (+) \vee \Box_b (+)) \wedge (\Diamond_a (-) \vee \Diamond_b (-))$).

  Moreover, assuming seriality of the models, the following bundles also become PN-independent:
  $(\Box (+) \wedge (-)) \vee (\Box (-) \wedge (+))$ (radical ignorance \cite{VF-2020-WHRI,JF-2025-LFBRB}; equivalent to $((+) \vee \Box (+)) \wedge ((-) \vee \Box (-))$ over serial models);
  $(\Box_a (-) \wedge \Box_b (+)) \vee (\Box_a (+) \wedge \Box_b (-))$ (strong disagreement \cite{TP-2018-SBD}; equivalent to $(\Box_a (+) \vee \Box_b (+)) \wedge (\Box_a (-) \vee \Box_b (-))$ over serial models).

  On the other hand, one non-ad hoc example of non-PN-independent bundles is the so-called \emph{Rumsfeld ignorance} or \emph{unknown unknowns} $\Diamond (+) \wedge \Diamond (-) \wedge \Diamond (\Box(+) \vee \Box(-))$ \cite{KF-2018-IOI, JF-2023-AIRI}:
  both itself and its dual are not PN-independent,
  even over models with rather strong properties (see Appendix \ref{appd.ncvri}).
\end{example}

For PN-independent bundles, we can also define a stronger notion of bisimulation,
which does not quantify over \emph{all} subsets of the $\tau$-domains,
but only the neighborhoods generated by $\tau$.
Compared to the kind of neighborhood bisimulation defined in \cite{HH-2007-NBHBIS},
our bisimulation for PN-independent bundles shows greater resemblance to 
the bisimulation for monotonic neighborhood models (see e.g.\ \cite{EP-2017-NBHSM} for the definition)
and the bisimulation for (non-)contingency modality defined in \cite{JF-2014-AN}.

\begin{definition}[Independent Bisimuluation]\label{def.cvbis}
  Let $\M$ be a model, and let $\tau \in \TT_\BA$.
  Let $\Z$ be a binary relation on $\M$.
  We say $\Z$ is an \emph{independent $\tau$-bisimulation} on $\M$,
  if for all $(w, v) \in \Z$, we have
  \begin{itemize}
    \item (Inv) For all $p \in \mathbf{P}$, $w \in V(p)$ iff $v \in V(p)$;
    \item ($\tau$-Zig) For all $(X_0, X_1) \in \Nbh_\M(w, \tau)$ s.t.\ $\Z[X_0] \cap X_1 = \emptyset$,
    there is $(Y_0, Y_1) \in \Nbh_\M(v, \tau)$ s.t.\ 
    $(Y_0, Y_1) \sqsubseteq (\Z[X_0], \Z[X_1])$;
    \item ($\tau$-Zag) For all $(Y_0, Y_1) \in \Nbh_\M(v, \tau)$ s.t.\ $\Z[Y_0] \cap Y_1 = \emptyset$,
    there is $(X_0, X_1) \in \Nbh_\M(w, \tau)$ s.t.\ $(X_0, X_1) \sqsubseteq (\Z^{-1}[Y_0], \Z^{-1}[Y_1])$.
\end{itemize}

We say $\M,w$ and $\N,v$ are \emph{independently $\tau$-bisimilar},
denoted by $\M,w \bis_\tau^\times \N,v$,
if there is an independent $\tau$-bisimulation $\Z$ on the disjoint union of $\M$ and $\N$ s.t.\ $(w, v) \in \Z$.
\end{definition}

Note that when $\tau^+ \in \TT^+_\BA$ (so that for any $w$ in $\M$, $\Nbh_\M(w, \tau^+) = N^+_{w, \tau^+} \times \{\emptyset\}$ for some $N^+_{w, \tau^+}$
),
for any relation $\Z$ on $\M$ and any $(X, Y) \in \Nbh_\M(w, \tau^+)$,
we trivially have $\Z[X] \cap Y = \emptyset$,
so ($\tau^+$-Zig) and ($\tau^+$-Zag) in Definition \ref{def.cvbis} can be simplified as follows:
\begin{itemize}
    \item ($\tau^+$-Zig) For all $(X, \emptyset) \in \Nbh_\M(w, \tau^+)$
    there is $(Y, \emptyset) \in \Nbh_\M(v, \tau^+)$ s.t.\ 
    $Y \subseteq \Z[X]$;
    \item ($\tau^+$-Zag) For all $(Y, \emptyset) \in \Nbh_\M(v, \tau^+)$,
    there is $(X, \emptyset) \in \Nbh_\M(w, \tau^+)$ s.t.\ $X \subseteq \Z^{-1}[Y]$.
\end{itemize}
In other words, independent bisimulation for $\tau^+ \in \TT^+_\BA$ is essentially the bisimulation for \emph{monotonic neighborhood models} \cite{EP-2017-NBHSM} applied to $\tau^+$-generated neighborhood structures.

Also note that for the contingency bundle $\tau_\triangledown = \Diamond (+) \wedge \Diamond (-)$,
$\Nbh_\M(w, \tau_\triangledown) = \{\{v\} \mid v \in R(w)\}^2$ for any $\M,w$,
so the conditions ($\tau_\triangledown$-Zig) and ($\tau_\triangledown$-Zag) can be simplified as follows:
\begin{itemize}
    \item ($\tau_\triangledown$-Zig) For all $w_0, w_1 \in R(w)$ s.t.\ $(w_0, w_1) \notin \Z$,
    there are $v_0, v_1 \in R(v)$ s.t.\ 
    $(w_0, v_0), (w_1, v_1) \in \Z$;
    \item ($\tau_\triangledown$-Zag) For all $v_0, v_1 \in R(w)$ s.t.\ $(v_0, v_1) \notin \Z$,
    there are $w_0, w_1 \in R(v)$ s.t.\ 
    $(w_0, v_0), (w_1, v_1) \in \Z$.
\end{itemize}
Then, what we obtain is essentially the bisimulation for (non-)contingency defined in \cite{JF-2014-AN}.

Again, it is rather routine to check that $\bis^\times_\tau$ is an equivalence relation.
The reader could also check that $\M,w_0$ and $\N,v_0$ given in Example \ref{ex.bis} are independently $\tau$-bisimilar.
We then show that $\bis^\times_\tau$ always coincides with $\bis_\tau$ when $\tau$ is PN-independent (proof in Appendix \ref{appd.cvbis}).

\begin{proposition}\label{prop.cvbis}
  For any models $\M,w$, $\N,v$ and $\tau \in \TT_\BA$,
  $\M,w \bis^\times_\tau \N,v \implies \M,w \bis_\tau \N,v$;
  moreover, if $\tau$ is PN-independent over $\M$ and $\N$,
  then $\M,w \bis_\tau \N,v \implies \M,w \bis^\times_\tau \N,v$.
\end{proposition}

\begin{corollary}
  For any models $\M,w$, $\N,v$ and $\tau \in \TT_\BA$, $\M,w \bis^\times_\tau \N,v \implies \M,w \equiv_\tau \N,v$;
  moreover,
  if $\M$, $\N$ are also image-finite and $\tau$ is PN-independent over $\M$ and $\N$,
  then $\M,w \bis^\times_\tau \N,v \iff \M,w \equiv_\tau \N,v$.
\end{corollary}

\subsection{Convex Neighborhood Semantics}

Next, we consider axiomatizations of PN-independent bundled modalities.
Our strategy is to first consider a kind of \emph{convex neighborhood semantics} for $\LL_\bund$,
which abstract the core structure of the bundled semantics for all PN-independent bundles,
and then transfer completeness results under such a semantics to concrete bundled semantics on Kripke models using representation theorems.

\begin{definition}[Convex Neighborhood Models]
  A \emph{convex neighborhood model} is a $4$-tuple $\N = $\\ $(W, N^+, N^-, V)$,
  where $W \neq \emptyset$ is the set of worlds,
  $N^+$, $N^-$ are both functions from $W$ to $\wp(\wp(W))$,
  and $V: \BP \to \wp(W)$ is the valuation function.
  A convex neighborhood frame is a model without the $V$ part.

  We say $\N$ is a \emph{core model} if for all $w \in W$, we have 
  (i) for all $X \in N^+(w)$, there is $Y \in N^-(w)$ s.t.\ $X \cap Y = \emptyset$;
  (ii) for all $Y \in N^-(w)$, there is $X \in N^+(w)$ s.t.\ $X \cap Y = \emptyset$, and 
  (iii) for all $\star \in \{+, -\}$ and $Z, Z' \in N^\star(w)$,
  if $Z \subseteq Z'$, then $Z = Z'$.\footnote{
    Our core convex neighborhood models generalize core monotonic neighborhood models \cite{HH-2003-MML,EP-2017-NBHSM}.
  }
\end{definition}

Then, we define the following kind of \emph{support semantics} on convex neighborhood models:

\begin{definition}[Neighborhood Semantics]
  Let $\N = (W, N^+, N^-, V)$ be a convex neighborhood model.
  For any $w \in W$ and $\phi \in \LL_\bund$, the satisfaction relation $\Vdash$ is defined recursively as usual, with the clause for $\bund$ being
  ($|\phi|_\N = \{w \in W \mid \N,w \Vdash \phi\}$):
  \begin{center}
  \begin{tabular}{|lcl|}
  \hline
  &&\\[-0.9em]
  $\N,w \Vdash \bund \phi$ & $\iff$ & $\exists X \in N^+(w) \exists Y \in N^-(w) ((X, Y) \sqsubseteq (|\phi|_\N, W \setminus |\phi|_\N))$.\\[0.1em]
  \hline
  \end{tabular}
  \end{center}
  As usual, validity on frames is defined as truth at all worlds under all valuations.
\end{definition}

Equivalently, given a convex neighborhood model $\N = (W, N^+, N^-, V)$, 
the above semantics can also be defined by first generating a neighborhood function $N$ from $N^+$ and $N^-$,
and then applying the exact neighborhood semantics to this $N$ (i.e.\  $\M,w \Vdash \bund \phi$ iff $|\phi|_\N \in N(w)$),
where the $N$ in question is defined by $N(w) = \{Z \subseteq W \mid \exists X \in N^+(w) \exists Y \in N^-(w) ((X, Y) \sqsubseteq (Z, W \setminus Z)\}$ for all $w \in W$.
Conversely, given an ordinary neighborhood model $\N = (W, N, V)$ where $N$ maps each world to a \emph{convex} family of neighborhoods,
we can also generate a pair of positive and negative neighborhood functions $N^+$ and $N^-$ by letting $N^+ = N$ and $N^-(w) = \{W \setminus X \mid X \in N(w)\}$ for all $w \in W$;
then, the exact neighborhood semantics w.r.t.\ $N$ also coincides with the above support semantics w.r.t.\ $N^+$ and $N^-$ (i.e.\ $|\phi|_\N \in N(w)$ iff $\exists X \in N^+(w)\exists Y \in N^-(w) ((X, Y) \sqsubseteq (|\phi|_\N, W \setminus |\phi|_\N))$).
The situation here resembles monotonic neighborhood logic,
where we also have an equivalence between the support semantics w.r.t.\ an arbitrary neighborhood function $N^+$ and the exact semantics w.r.t.\ its upward closure $N$
(i.e.\ $|\phi|_\N \in N(w)$ iff $\exists X \in N^+(w) X \subseteq |\phi|_\N$, where $N(w) = \{Z \subseteq W \mid \exists X \in N^+(w) X \subseteq Z\}$),
and the two kinds of semantics coincide if the neighborhood function in question is upward closed.

Next, by a \emph{convex modal logic}, we mean a set of $\LL_\bund$-formulas that contains all instances of the axioms and is closed under the rules displayed in the following table as well as uniform substitution.
The minimal convex modal logic is $\mathbf{CONV}$.\footnote{
  It is worth noting that in the literature, 
  the axiom $\mathtt{CONV}$ sometimes appears in the form of the rule 
  that allows one to infer $\bund \psi \wedge \bund \chi \to \bund \phi$ from $\psi \to \phi$ and $\phi \to \chi$ (see \cite{TP-2018-SBD, TP-2019-MBD, ZX-2022-SRI} for example).
  This is derivable in $\mathbf{CONV}$, since if $\vdash \psi \to \phi$, then $\vdash \psi \leftrightarrow (\phi \land \psi)$, and if $\vdash \phi \to \chi$, then $\vdash \chi \leftrightarrow (\phi \lor \chi)$.
  When we cite $\mathtt{CONV}$, we occasionally mean this rule.}
\begin{center}
  \begin{tabular}{|ll|ll|}
    \hline
    \multicolumn{4}{|c|}{Shared axioms and rules of convex modal logics}\\
    \hline
    \multicolumn{2}{|c|}{\textbf{Axioms}} & \multicolumn{2}{|c|}{\textbf{Rules}}\\
    $\mathtt{TAUT}$ & All propositional tautologies & 
    $\mathtt{MP}$ & From $\phi$ and $\phi \to \psi$, infer $\psi$\\
    $\mathtt{CONV}$ & $\bund (\phi \wedge \psi) \wedge \bund (\phi \vee \chi) \to \bund \phi$ &
    $\mathtt{RE}$ & From $\phi \leftrightarrow \psi$, infer $\bund \phi \leftrightarrow \bund \psi$\\[0.1em]
    \hline
  \end{tabular}
\end{center}

\begin{proposition}
$\mathbf{CONV}$ is sound w.r.t.\ the class of all convex neighborhood models (frames).
\end{proposition}

For the remainder of this section, we construct canonical models for convex modal logics. 
Note that the canonical models we define are not the usual kind of canonical models for neighborhood logics in general;
they should rather be viewed as a generalization of the kind of canonical models for \emph{monotonic} neighborhood logics defined in \cite{YD-2023-SKWA}.
An important feature of our canonical models (as well as the canonical models in \cite{YD-2023-SKWA}) is that the canonical models are automatically \emph{core} models,
and correspondingly, such canonical models admit a greater amount of canonical axioms
compared to the standard canonical models for neighborhood logics in general.

\begin{definition}[Core]
  Let $\BL$ be a convex modal logic.
  An \emph{$\BL$-filter} is a subset $F$ of $\LL_\bund$ s.t.\ 
  (i) $\top \in F$, 
  (ii) $\phi, \psi \in F$ implies $\phi \wedge \psi \in F$, and  
  (iii) $\phi \in F$, $\vdash_\BL \phi \to \psi$ implies $\psi \in F$.
  Note that $\LL_\bund$ is a filter.

  Given a set $\Delta$ of $\LL_\bund$-formulas, we say an $\BL$-filter $F$ is $+$-\emph{coherent} (resp.\ $-$-\emph{coherent}) with $\Delta$
  if for all $\phi \in F$, there is $\psi \in F$ s.t.\ $\bund (\phi \wedge \psi) \in \Delta$ (resp.\ $\nbund (\phi \wedge \psi) \in \Delta$).\footnote{Equivalently, an $\BL$-filter $F$ is $+$-coherent with $\Delta$ if $\{\phi \in F \mid \bund \phi \in \Delta\}$ is \emph{dense} in $F$ in the sense that for any $\phi \in F$ there is $\psi \in F$ with $\vdash_\BL \psi \to \phi$ such that $\bund\psi \in \Delta$; similarly for $-$-coherence with $\Delta$.}
  Let $\PC_\BL(\Delta)$ (resp.\ $\NC_\BL(\Delta)$) denote the set of \emph{maximal} $\BL$-filters among those $+$-coherent (resp.\ $-$-coherent) with $\Delta$.
\end{definition}

\begin{definition}[Canonical Models]
  Let $\BL$ be a convex modal logic.
  The canonical model w.r.t.\ $\BL$ is $\N^c_\BL = (W^c_\BL, N^\oplus_\BL, N^\ominus_\BL, V^c_\BL)$, where 
  $W^c_\BL = \{\Delta \subseteq \LL_\bund \mid \Delta \text{ is a maximally $\BL$-consistent set}\}$,
  $V^c(p) = \{\Delta \in W^c_\BL \mid p \in \Delta\}$ for all $p \in \BP$,
  and for any $\Delta \in W^c_\BL$, $N^\oplus_\BL(\Delta)$ and $N^\ominus_\BL(\Delta)$ are defined as follows:
  \begin{center}
    $\ext{F} = \{\Theta \in W^c_\BL \mid F \subseteq \Theta\}$;\,
    $N^\oplus_\BL(\Delta) = \{\ext{F} \mid F \in \PC_\BL(\Delta)\}$;\,
    $N^\ominus_\BL(\Delta) = \{\ext{F'} \mid F' \in \NC_\BL(\Delta)\}$.
  \end{center}
\end{definition}

Intuitively, if a filter $F$ is $+$-coherent (resp.\ $-$-coherent) w.r.t.\ $\Delta$,
then it is consistent with the information in $\Delta$ to have a positive (resp.\ negative) neighborhood supporting all formulas in $F$.
When defining the canonical models, 
we want to realize the formulas in an MCS $\Delta$ with a \emph{minimal} number of positive and negative neighborhoods, so we want each neighborhood to support as many formulas as possible;
thus, the positive (resp.\ negative) neighborhoods of $\Delta$ in the canonical models are constructed from the \emph{maximal} filters among the ones that are $+$-coherent (resp.\ $-$-coherent) with $\Delta$.

Now, we fix an arbitrary convex modal logic $\BL$. 
To avoid notational clutter, 
from now on,
we omit all prefixes and subscripts ``$\mathbf{L}$''.
The following shows that the canonical model w.r.t.\ $\BL$ indeed has all the properties we need.

\begin{lemma}\label{lem.ext}
  For any $\Delta \in W^c$, if $F_0$ is a filter $+$-coherent (resp.\ $-$-coherent) with $\Delta$,
  then there is $F \in \PC(\Delta)$ (resp.\ $\NC(\Delta)$) s.t.\ $F_0 \subseteq F$.
  Thus, for any $\bund \phi \in \Delta$,
  there is $F \in \PC(\Delta)$ and $F' \in \NC(\Delta)$ s.t.\ $\phi \in F$ and $\neg \phi \in F'$.
\end{lemma}

\begin{proof}
  Observe that if $\mathcal{F}$ forms a $\subseteq$-chain of filters $+$-coherent (resp.\ $-$-coherent) with $\Delta$,
  then $\bigcup \mathcal{F}$ is also a filter that is $+$-coherent (resp.\ $-$-coherent) with $\Delta$.
  Thus, by Zorn's Lemma, every $+$-coherent (resp.\ $-$-coherent) filter can be extended into a maximal one.\footnote{
  Since the language $\LL_\bund$ is countable,
  our use of Zorn's Lemma here is not essential.
  For example, we could enumerate all formulas in $\LL_\bund$,
  add them one by one to the filter in question if admissible, and repeat the procedure for countably many times.
  }
  If $\bund \phi \in \Delta$, then the filter $F_0$ generated by $\{\phi\}$ consisting of all logical consequences of $\phi$ under $\BL$ is $+$-coherent with $\Delta$, since for any $\psi \in F_0$, $\vdash \phi \leftrightarrow (\psi \land \phi)$ and thus $\bund (\psi \land \phi) \in \Delta$.
  Similarly, the filter generated by $\{\lnot\phi\}$ is $-$-coherent with $\Delta$.
\end{proof}

\begin{lemma}\label{lem.nwitt}
  For any $\Delta \in W^c$, if $\bund \phi \notin \Delta$,
  then there is no $F \in \PC(\Delta)$ and $F' \in \NC(\Delta)$ satisfying $\phi \in F$ and $\neg \phi \in F'$.
\end{lemma}

\begin{proof}
  Note that if there is $F \in \PC(\Delta)$ and $F' \in \NC(\Delta)$ s.t.\ $\phi \in F$ and $\neg \phi \in F'$,
  then since $F$ and $F'$ are $+$-coherent and $-$-coherent with $\Delta$ respectively, there is $\psi \in F$ and $\chi \in F'$
  s.t.\ $\bund (\phi \wedge \psi), \nbund (\neg \phi \wedge \chi) \in \Delta$.
  Also note that $\nbund (\neg \phi \wedge \chi) \in \Delta$ implies $\bund (\phi \vee \neg \chi) \in \Delta$ by $\mathtt{RE}$.
  Thus $\bund \phi \in \Delta$ by $\mathtt{CONV}$.
\end{proof}

It is worth noting that in our proof for the completeness of the minimal convex modal logic $\mathbf{CONV}$,
Lemma \ref{lem.nwitt} is the only place where the axiom $\mathtt{CONV}$ is used.

% With the help of the above lemmas, we can prove the following \emph{Truth Lemma}.
\begin{proposition}\label{prop.core}
  $\N^c$ is a core model.
\end{proposition}

\begin{proof}
  Let $\Delta \in W^c$ be arbitrary.
  For requirement (i), take any $F \in \PC(\Delta)$.
  Since $F$ is $+$-coherent, there is $\phi \in F$ such that $\bund\phi \in \Delta$.
  So by Lemma \ref{lem.ext}, there is $F' \in \NC(\Delta)$ s.t.\ $\neg \phi \in F'$,
  and thus $\ext{F} \cap \ext{F'} = \emptyset$.
  Requirement (ii) can be proved similarly.
  For (iii), take $F, G \in \PC(\Delta)$ s.t.\ $\ext{F} \subseteq \ext{G}$.
  Now $F$ must extend $G$, since otherwise there would have been a $\phi \in G \setminus F$,
  and $F \cup \{\lnot\phi\}$ can be extended to a maximally consistent set $\Theta \in W^c$ such that $\Theta \in \ext{F} \setminus \ext{G}$, contradicting $\ext{F} \subseteq \ext{G}$.
  But then $G = F$ since $F$ is maximal among those filters $+$-coherent with $\Delta$, of which $G$ is a member.
  Hence, $\ext{F} = \ext{G}$.
  The case for $F',G' \in \NC(\Delta)$ is similar.
\end{proof}

\begin{lemma}[Truth]\label{lem.truth}
  For any $\Delta \in W^c$ and $\phi \in \LL_\bund$, $\N^c, \Delta \Vdash \phi \iff \phi \in \Delta$.
\end{lemma}

\begin{proof}
  As usual, the lemma is proved by induction. We only consider the inductive case for $\bund \phi$.

  By Lemma \ref{lem.ext} and \ref{lem.nwitt}, 
  $\bund \phi \in \Delta$ iff there is $F \in \PC(\Delta)$ and $F' \in \NC(\Delta)$ s.t.\ $\phi \in F$ and $\neg \phi \in F'$.
  Moreover, by IH, $\phi \in F$ iff for any $\Theta \in \ext{F}$, $\N^c, \Theta \Vdash \phi$,
  and $\neg \phi \in F'$ iff for any $\Theta \in \ext{F'}$, $\N^c, \Theta \not \Vdash \phi$.
  Thus, $\bund \phi \in \Delta$ iff there is $F \in \PC(\Delta)$ and $F' \in \NC(\Delta)$ s.t.\ $(\ext{F}, \ext{F'}) \sqsubseteq (|\phi|_{\N^c}, W^c \setminus |\phi|_{\N^c})$ iff $\N^c,\Delta \Vdash \bund \phi$.
\end{proof}

Thus, when the logic $\BL$ in question is just $\mathbf{CONV}$, we obtain the following completeness theorem:
\begin{theorem}
  Under convex neighborhood semantics, $\mathbf{CONV}$ is sound and strongly complete w.r.t.\ the class of all models,
  and also the class of all core models.
\end{theorem}

\section{Case Studies}\label{sec.cases}
In this section, we show how the above canonical model construction could help us achieve completeness results for several concrete PN-independent bundled modalities 
over classes of models with rather strong properties.
To our knowledge, none of the axiomatizations presented below have appeared in the literature.
The following notion of \emph{representation} is key in linking neighborhood completeness to Kripke completeness under bundled semantics.

\begin{definition}[Representation]\label{def.rep}
  Let $\N = (W, N^+, N^-, V)$ be a convex neighborhood model,
  $\M = (W^*, \{R^*_a\}_{a \in \BA}, V^*)$ be a Kripke model,
  and let $\tau \in \TT_\BA$, $v \in W$, $w \in W^*$.
  We say $\N,v$ is \emph{$\tau$-represented} by $\M,w$,
  if there is a function $\pi: W^* \to W$ s.t.\ 
  $\pi(w) = v$,
  and the following are satisfied for all $u \in W^*$:
  \begin{itemize}
    \item (Inv) For any $p \in \BP$, $u \in V^*(p)$ iff $\pi(u) \in V(p)$;
    \item (Forth) For any $(X, Y) \in \Nbh_\M(u, \tau)$ s.t.\ $\pi[X] \cap \pi[Y] = \emptyset$,
    there is $X' \in N^+(\pi(u))$ and $Y' \in N^-(\pi(u))$
    s.t.\ $(X', Y') \sqsubseteq (\pi[X], \pi[Y])$;
    \item (Back) For any $X' \in N^+(\pi(u))$ and $Y' \in N^-(\pi(u))$ s.t.\ $X' \cap Y' = \emptyset$,
    there is $(X, Y) \in \Nbh_\M(u, \tau)$ s.t.\ 
    $(\pi[X], \pi[Y]) \sqsubseteq (X', Y')$.
  \end{itemize}
\end{definition}

With the help of Proposition \ref{prop.nbhsem1}, it is rather routine to check the following by induction:

\begin{proposition}
  If a convex neighborhood model $\N, v$ is $\tau$-represented by a Kripke model $\M, w$,
  then for any $\phi \in \LL_\bund$, $\N,v \Vdash \phi \iff \M,w \vDash_\tau \phi$.
\end{proposition}

\subsection{Someone Knows and Disagreement within Group}
In this subsection, we fix a non-empty finite set of indices $A \subseteq \rho^{-1}[\{1\}]$, and consider the ``someone knows'' bundle $\tau_\mathtt{SK}(A) = \bigvee_{a \in A} \Box_a (+)$ studied in \cite{AT-2021-SK, YD-2023-SKWA}
and the bundle $\tau_\mathtt{DG}(A) = \bigvee_{a, b \in A} \Box_a (+) \wedge \Box_b (-)$ for disagreement among a group of agents based on $A$,
where the latter generalizes the bundle for disagreement between two agents studied in \cite{TP-2018-SBD}.
We focus on axiomatizations of the $\tau_\mathtt{SK}(A)$-bundle on $S5$-models and
the $\tau_\mathtt{DG}(A)$-bundle on $KD45$-models.

Table \ref{table.ax1} lists axioms and properties of convex neighborhood models relevant to our axiomatizations.
It is routine to check that for any convex neighborhood frame,
if it has any property listed in the table,
then the corresponding axiom is valid over it.
Moreover, we can show that these axioms correspond to the properties \emph{canonically}
(the proof of the following lemma is in Appendix \ref{appd.cano}).

\begin{table}
\centering
  {\small
  \renewcommand\arraystretch{1.2}
  \begin{tabular}{|ll|l|}
    \hline
    \multicolumn{2}{|l|}{\textbf{Axioms}} & \textbf{Properties}\\
    \hline
    $\mathtt{EQU}$ & $\bund \phi \to \nbund \phi$ & 
    (Sym) \, $N^+_w = N^-_w$\\
    $\mathtt{N}_\bund$ & $\bund \top$ & 
    (Pur)$^+$ \, $N^+_w \neq \emptyset \wedge \emptyset \in N^-_w$ \\
    $\mathtt{D}_\bund$ & $\neg \bund \bot$ & 
    (Ser)$^+$ \, $\emptyset \notin N^+_w$\\
    $\mathtt{T}_\bund$ & $\bund \phi \to \phi$ & 
    (Refl)$^+$ \, $\forall X \in N^+_w \, w \in X$ \\
    $\mathtt{C}^n_\bund$ & $\bigwedge_{0 \leq i \leq n} \bund \phi_i \to \bigvee_{0 \leq i < j \leq n} \bund (\phi_i \wedge \phi_j)$ & 
    ($n$-d-Bd)$^+$ \, $|N^+_w| \leq n$\\
    $\mathtt{4}'_\bund$ & $\bund \phi \to \bund (\phi \wedge (\bund (\phi \vee \psi) \to \bund \phi))$ &
    (PI$'$)$^+$ \,\! $\forall X \in N^+_w \forall v \in X ((\exists Y \in N^-_v X \cap Y = \emptyset) \Rightarrow X \in {\uparrow} N^+_v)$\\
    \multirow{2}{*}{$\mathtt{5C}^n_\bund$} & $\bund \phi \wedge \bigwedge_{0 \leq i \leq n} \neg \bund (\phi \wedge \psi_i) \to \bund (\phi \wedge $ & 
    \multirow{2}{*}{($n$-Coa)$^+$ \, $\forall X \in N^+_w \forall v \in X |N^+_v \setminus {\uparrow} \{X\}| \leq n$} \\[-0.4em]
    & $(\bigwedge_{0 \leq i \leq n} \bund \psi_i \to \bigvee_{0 \leq i < j \leq n} \bund (\psi_i \wedge \psi_j)))$ &\\
    \hline
    \multicolumn{2}{|l|}{$\mathtt{N}_\nbund$, $\mathtt{D}_\nbund$, $\mathtt{T}_\nbund$, $\mathtt{C}^n_\nbund$, $\mathtt{4}'_\nbund$, $\mathtt{5C}^n_\nbund$
    are axioms obtained} & 
    (Pur)$^-$, (Ser)$^-$, (Refl)$^-$, ($n$-d-Bd)$^-$, (PI$'$)$^-$, ($n$-Coa)$^-$ are\\
    \multicolumn{2}{|l|}{by replacing $\bund$ with $\nbund$ in the original axioms.} &
    defined by switching $+$ and $-$ in the original definitions.\\
    \hline
  \end{tabular}
  }

  \caption{The first set of axioms and corresponding properties.
  In definitions of the properties,
  we always omit the first universal quantifier $\forall w \in W$; 
  we also use $N_w^\star$ to denote $N^\star(w)$ for any $\star \in \{+, -\}$,
  and ${\uparrow} N$ to denote $\{Z \subseteq W \mid \exists X \in N \, X \subseteq Z\}$.}
  \label{table.ax1}
\end{table}

\begin{lemma}\label{lem.cano1}
  For any convex modal logic $\BL$ and any axiom $\alpha$ in Table \ref{table.ax1},
  if all instances of $\alpha$ is in $\BL$, then the canonical model of $\BL$ has the property corresponding to $\alpha$.
\end{lemma}

We then define the following two particular systems using axioms in Table \ref{table.ax1}:
\begin{align*}
\mathbf{SKS5}(A) := & \; \mathbf{CONV} \oplus \{\mathtt{N}_\bund, \mathtt{T}_\bund, \mathtt{C}^{|A|}_\bund, \mathtt{4}'_\bund, \mathtt{5C}^{|A|-1}_\bund\}\\
\mathbf{DGD45}(A) := & \; \mathbf{CONV} \oplus \{\mathtt{EQU}, \mathtt{D}_\bund, \mathtt{C}^{|A|}_\bund, \mathtt{4}'_\bund, \mathtt{5C}^{|A|-1}_\bund\}.
\end{align*}

By the canonical model construction in the previous section and Lemma \ref{lem.cano1}, we easily obtain the following neighborhood completeness results:

\begin{lemma}\label{lem.nbhcom1}
  Under convex neighborhood semantics,
  $\mathbf{SKS5}(A)$ is strongly complete w.r.t.\ the class of core models satisfying 
  (Pur)$^+$, (Refl)$^+$, ($|A|$-d-Bd)$^+$, (PI$'$)$^+$, ($(|A| - 1)$-Coa)$^+$,
  and $\mathbf{DGD45}(A)$ is strongly complete w.r.t.\ the class of core models satisfying 
  (Sym), (Ser)$^+$, ($|A|$-d-Bd)$^+$, (PI$'$)$^+$, ($(|A| - 1)$-Coa)$^+$.
\end{lemma}

Next, in order to turn the above neighborhood completeness into Kripke completeness under the corresponding bundled semantics,
it suffices to show that neighborhood models with the sets of properties listed in Lemma \ref{lem.nbhcom1} 
can be represented by the corresponding kinds of Kripke models.

\begin{lemma}\label{lem.unr}
  Let $\N, w$ be a core model, and assume that $\N$ satisfies ($|A|$-d-Bd)$^+$, (PI$'$)$^+$ and ($(|A| - 1)$-Coa)$^+$.
  If $\N$ also satisfies (Pur)$^+$ and (Refl)$^+$,
  then $\N$ can be $\tau_\mathtt{SK}(A)$-represented by an $S5$-Kripke model;
  on the other hand, if $\N$ also satisfies (Sym) and (Ser)$^+$,
  then $\N$ can be $\tau_\mathtt{DG}(A)$-represented by a $KD45$-Kripke model.
\end{lemma}

\begin{proof}[Proof Sketch]
  Let $\N = (W, N^+, N^-, V)$ be a core model satisfying ($|A|$-d-Bd), (PI$'$)$^+$ and ($(|A|-1)$-Coa)$^+$,
  and let $w \in W$ be arbitrary.
  When $N^+(w) = \emptyset$, the construction is trivial,
  so we assume that $N^+(w) \neq \emptyset$.
  First, as preparation, we show that for any $u \in W$, $X \in N^+(u)$ and $v \in X$,
  $|N^+(v) \cup \{X\}| \leq |A|$.
  We consider two cases:
  \begin{itemize}
    \item Case 1: If $X \notin {\uparrow}N^+(v)$, then by (PI$'$)$^+$,
    for any $Y \in N^-(v)$, $X \cap Y \neq \emptyset$.
    Then, since $\N$ is a core model, there cannot be any $X' \in N^+(v)$ s.t.\ $X \subseteq X'$
    (for otherwise there should also be $Y' \in N^-(v)$ s.t.\ $X' \cap Y' = \emptyset$, and thus $X \cap Y' = \emptyset$).
    Hence, by ($(|A|-1)$-Coa)$^+$, $|N^+(v)| < |A|$,
    and thus $|N^+(v) \cup \{X\}| \leq |A|$.
    \item Case 2: If $X \in {\uparrow}N^+(v)$, then there is $Z \in N^+(v)$ s.t.\ $Z \subseteq X$.
    Then, since $\N$ is a core model, there is no $X' \in N^+(v)$ s.t.\ $Z \subseteq X \subsetneq X'$,
    so $|N^+(v) \setminus \{X\}| = |N^+(v) \setminus {\uparrow}\{X\}| < |A|$ by ($(|A|-1)$-Coa)$^+$,
    and thus $|N^+(v) \cup \{X\}|\leq |A|$.
  \end{itemize}
  
  Then, we construct the Kripke model $\M$ using a kind of unraveling.
  We first recursively define a countable sequence $(W_k, \{R^k_a\}_{a \in A})_{k \in \Nat}$
  s.t.\ $W_k \subseteq (A \times W)^{k+1}$,
  $R_a^k$ is a binary relation s.t.\
  $\text{dom}(R_a^k) \subseteq W_k$ and 
  $W_{k+1} \subseteq \bigcup_{a \in A} \text{ran}(R_a^k)$,
  and for all $\xi \in W_k$, 
  $R_a^k(\xi) \cap W_{k+1} \subseteq \{(\xi, a)\} \times W$
  and $\mathtt{end}[R_a^k(\xi)] \in \bigcup_{u \in W} N^+(u)$,
  where $\mathtt{end}((\zeta, a, v)) = v$ for all $(\zeta, a, v) \in \bigcup_{k \in \Nat} (A \times W)^{k+1}$. 
  \begin{itemize}
    \item Let $W_0 = \{(c, w)\}$ for some random $c \in A$.
    By ($|A|$-d-Bd)$^+$, $|N^+(w)| \leq |A|$,
    so we may index $N^+(w)$ with $A$ as $\{X_a\}_{a \in A}$.
    Then, for each $a \in A$, let 
    $R^0_a((c, w)) = \{(c, w, a, v) \mid v \in X_a\} \cup \{(c, w)\}$ if $w \in X_a$,
    and $R^0_a((c, w)) = \{(c, w, a, v) \mid v \in X_a\}$ otherwise.
    Let $W_1 = \{(c, w, a, v) \mid a \in A, v \in X_a\}$.
    \item Assume that $W_k$, $\{R^k_a\}_{a \in A}$, and $W_{k+1}$ are constructed.
    For all $\xi \in W_k$, $a \in A$ and 
    $\zeta = (\xi, a, v) \in R_a^k(\xi)$,
    let $X = \mathtt{end}[R^k_a(\xi)]$;
    then, since $v \in X \in \bigcup_{u \in W} N^+(u)$, as we proved above, 
    $|N^+(v) \cup \{X\}| \leq |A|$,
    so we may index $N^+(v) \cup \{X\}$ with $A$
    as $\{X_b\}_{b \in A}$, where $X_a = X$.
    Then, let $R_a^{k+1}(\zeta) = R_a^k(\xi)$,
    and for each $b \in A \setminus \{a\}$,
    let $R_b^{k+1}(\zeta) = \{(\zeta, b, u) \mid u \in X_b\} \cup \{\zeta\}$ if
    $v \in X_b$,
    and $R_b^{k+1}(\zeta) = \{(\zeta, b, u) \mid u \in X_b\}$ otherwise.
    Let $W_{k+2} = \{(\zeta, b, u) \mid b \in A \setminus \{a\}, u \in X_b\}$.
  \end{itemize}
  Then, let $\M = (W^*, \{R^*_a\}_{a \in A}, V^*)$,
  where $W^* = \bigcup_{k \in \Nat} W_k$,
  $R^*_a = \bigcup_{k \in \Nat} R^k_a$,
  and $V^*(p) = W^* \cap \mathtt{end}^{-1}[V(p)]$ for all $p \in \BP$.
  Note that $\M$ is by construction a $KD45$-model,
  and when $\N$ satisfies (Refl)$^+$,
  $\M$ is also an $S5$-model.
  Moreover, note that $\{\mathtt{end}[R^*_a((c, w))] \mid a \in A\} = N^+(w)$, 
  and for any $k>0$ and $\xi \in W_{k}$
  $\{\mathtt{end}[R^*_a(\xi)] \mid a \in A\} = N^+(\mathtt{end}(\xi)) \cup \{X\}$ for some $X \in \bigcup_{u \in W} N^+(u)$ s.t.\ $\mathtt{end}(\xi) \in X$.
  Thus, we could check that when $\N$ has the corresponding properties,
  $\mathtt{end}$ (restricted to $W^*$) satisfies the conditions listed in Definition \ref{def.rep} for $\tau_\mathtt{SK}(A)$ or $\tau_\mathtt{DG}(A)$,
  so $\N,w$ is $\tau_\mathtt{SK}(A)$-represented or $\tau_\mathtt{DG}(A)$-represented by $\M,(c,w)$
  since $\mathtt{end}((c, w)) = w$.
\end{proof}

Then, combining the lemmas above, we obtain the following completeness theorem.

\begin{theorem}
  $\mathbf{SKS5}(A)$ is sound and strongly complete w.r.t.\  the class of $S5$-models 
  under $\tau_\mathtt{SK}(A)$-bundled semantics,
  and $\mathbf{DGD45}(A)$ is sound and strongly complete w.r.t.\ the class of $KD45$-models 
  under $\tau_\mathtt{DG}(A)$-bundled semantics.
\end{theorem}

Using axioms in Table \ref{table.ax1}, we can also axiomatize even more complicated bundles,
say the bundle $\bigvee_{a \in A} B_a (+) \wedge \bigvee_{b \in B} K_b (-)$
expressing that someone in $A$ erroneously believes in something known to be false by someone else in $B$.
However, we shall not discuss such cases here due to limited space.

\subsection{Belief without Knowledge}

In this section, we consider the bundled modality $\tau_\mathtt{MB} = \Diamond \Box (+) \wedge \Diamond (-)$ over $S4.2$-models.
The class of $S4.2$-models is noteworthy since it validates a desirable logic of knowledge and belief 
when we take $\Box \phi$ as the knowledge modality $K \phi$,
and define belief $B \phi$ as $\Diamond \Box \phi$ (for more details, see \cite{WL-1978-EL, RS-2006-LBK} for example).
Then, $\bund \phi$ with the above bundled semantics expresses $B \phi \wedge \neg K \phi$, which is a kind of \emph{belief without knowledge} or \emph{mere belief}.
Note that on $S4.2$-models, 
$\bund \phi$ is also equivalent to \emph{Dunning-Kruger ignorance} $BK \phi \wedge \neg K \phi$ \cite{JF-2025-RIDKI},\footnote{
  In \cite{JF-2025-RIDKI}, a logic with both the bundle $\neg K \phi \wedge BK \phi$ and $K \phi$ is axiomatized over $S4.2$-models (which are viewed as epistemic-doxastic models validating certain interaction principles there).
  But our result here is novel since $K \phi$ is not in our language.}
and $\bund \phi \vee \nbund \phi$ expresses the aforementioned ``Rumsfeld ignorance''.

Table \ref{table.ax2} lists relevant axioms and properties that are not covered by Table \ref{table.ax1}.
Again, frames with the property corresponding to an axiom in Table \ref{table.ax2} validate the axiom.
These axioms also correspond to the properties canonically
(the proof is again in Appendix \ref{appd.cano}).

\begin{table}
\centering
  {\small
  \renewcommand\arraystretch{1.2}
  \begin{tabular}{|ll|l|}
    \hline
    \multicolumn{2}{|l|}{\textbf{Axioms}} & \textbf{Properties}\\
    \hline
    $\mathtt{DIV}_\bund^n$ & $\bund \bigvee_{0 \le i \leq n} \phi_i \to \bigvee_{0 \le i \leq n} \bund \bigvee_{j \neq i} \phi_j$ & 
    ($n$-s-Bd)$^+$ \, $\forall X \in N^+_w |X| \leq n$\\
    \multirow{2}{*}{$\mathtt{4}''_\bund$} & \multirow{2}{*}{$\bund \phi \wedge \bund (\phi \vee \psi) \to \bund (\phi \wedge \bund (\phi \vee \psi))$} & 
    \multirow{2}{*}{(PI$''$)$^+$} \, $\forall X \in N^+_w \forall v \in X (X \in {\uparrow}N^+_v \wedge$\\[-0.4em]
    & & \qquad \qquad \qquad \quad $\forall Y \in N^-_w (X \cap Y = \emptyset \Rightarrow Y \in {\uparrow}N^-_v))$ \\
    $\mathtt{4}'''_\bund$ & $\bund \phi \wedge \bund \psi \to \bund (\phi \wedge \bund \psi)$ & 
    (PI$'''$)$^+$ \, $\forall X \in N^+_w \forall v \in X (N^+_w \subseteq {\uparrow}N^+_v \wedge N^-_w \subseteq {\uparrow}N^-_v)$\\
    $\mathtt{5}'''_\bund$ & $\bund \phi \wedge \neg \bund \psi \to \bund (\phi \wedge \neg \bund \psi)$ & 
    (NI$'''$)$^+$ \, $\forall X \in N^+_w \forall v \in X (N^+_v \subseteq {\uparrow} N^+_w \wedge N^-_v \subseteq {\uparrow}N^-_w)$ \\
    $\mathtt{DE}_\bund$ & $\bund \phi \to \bund (\phi \wedge \neg \bund \psi)$ & 
    (Emp)$^+$ \, $\forall X \in N^+_w \forall v \in X (N^+_v = N^-_v = \emptyset)$\\
    $\mathtt{T}^\mathtt{op}_\bund$ & $\phi \wedge \bund (\phi \vee \psi) \to \bund \phi$ &
    (Id)$^+$ \, $(\exists Y \in N^-_w w \notin Y) \Rightarrow \{w\} \in {\uparrow} N^+_w $\\
    \hline
    \multicolumn{3}{|l|}{$\nbund$- and $-$-versions of the above axioms and properties are defined in the same way as Table \ref{table.ax1}.} \\
    \hline
  \end{tabular}
  }

  \caption{The second set of axioms and corresponding properties.
  Conventions for the abbreviations we use are the same as Table \ref{table.ax1}.}

  \label{table.ax2}
\end{table}

\begin{lemma}\label{lem.cano2}
  For any convex modal logic $\BL$ and any axiom $\alpha$ in Table \ref{table.ax2},
  if all instances of $\alpha$ is in $\BL$, then the canonical model of $\BL$ has the property corresponding to $\alpha$.
\end{lemma}

We define the following two systems using the axioms presented in Table \ref{table.ax1} and Table \ref{table.ax2}
(note that their only difference is that $\mathtt{4}''_\nbund$ is strengthened to $\mathtt{4}'''_\nbund$ in $\mathbf{MBS4F}$):
\begin{align*}
  \mathbf{MBS4.2} := & \; \mathbf{CONV} \oplus \{\mathtt{D}_\bund, \mathtt{C}^1_\bund, \mathtt{DE}_\bund, \mathtt{D}_\nbund, \mathtt{DIV}^1_\nbund, \mathtt{4}''_\nbund, \mathtt{5}'''_\nbund, \mathtt{T}^\op_\nbund\}\\
  \mathbf{MBS4F} := & \; \mathbf{CONV} \oplus \{\mathtt{D}_\bund, \mathtt{C}^1_\bund, \mathtt{DE}_\bund, \mathtt{D}_\nbund, \mathtt{DIV}^1_\nbund, \mathtt{4}'''_\nbund, \mathtt{5}'''_\nbund, \mathtt{T}^\op_\nbund\}.
\end{align*}

Again, using the canonical construction and Lemma \ref{lem.cano1}, Lemma \ref{lem.cano2},
we could obtain the following neighborhood completeness:

\begin{lemma}
  $\mathbf{MBS4.2}$ is strongly complete w.r.t.\ the class of core models 
  satisfying (Ser)$^+$, ($1$-d-Bd)$^+$, (Emp)$^+$, (Ser)$^-$, ($1$-s-Bd)$^-$, (PI$''$)$^-$, (NI$'''$)$^-$ and (Id)$^-$,
  and $\mathbf{MBS4F}$ is strongly complete w.r.t.\ the class of core models 
  satisfying (Ser)$^+$, ($1$-d-Bd)$^+$, (Emp)$^+$, (Ser)$^-$, ($1$-s-Bd)$^-$, (PI$'''$)$^-$, (NI$'''$)$^-$ and (Id)$^-$.
\end{lemma}

Then, the Kripke completeness of the above logics under bundled semantics follows from the representation result below.

\begin{lemma}
  Let $N,w$ be a core model s.t.\ $\N$ satisfies (Ser)$^+$, ($1$-d-Bd)$^+$, (Emp)$^+$,
  (Ser)$^-$, ($1$-s-Bd)$^-$, (PI$''$)$^-$, (NI$'''$)$^-$ and (Id)$^-$.
  Then, $\N,w$ can be $\tau_\mathtt{MB}$-represented by an $S4.2$-Kripke model;
  moreover, if $\N$ further satisfies (PI$'''$)$^-$,
  then $\N,w$ can be $\tau_\mathtt{MB}$-represented by an $S4F$-Kripke model.\footnote{
    $\M = (W, R, V)$ is an $S4F$-model if there is $X \subseteq W$ s.t.\ $R = (W \setminus X)^2 \cup (W \times X)$.
    In \cite{RS-2006-LBK}, such models are viewed as the \emph{maximal} epistemic extensions of $KD45$-doxastic models.
    An intuitive interpretation of such models can also be found in \cite{RS-2006-LBK}.
  }
\end{lemma}

\begin{proof}[Proof Sketch]
  Assume that $\N = (W, N^+, N^-, V)$ is a core model satisfying (Ser)$^+$, ($1$-d-Bd)$^+$, (Emp)$^+$,
  (Ser)$^-$, ($1$-s-Bd)$^-$, (PI$''$)$^-$, (NI$'''$)$^-$ and (Id)$^-$,
  and let $w \in W$ be arbitrary.
  When $N^+(w) = \emptyset$, the case is trivial,
  so we assume that $N^+(w) \neq \emptyset$.
  
  Since $N^+(w) \neq \emptyset$, by (Ser)$^+$ and ($1$-d-Bd)$^+$,
  there is a non-empty $X_w \subseteq W$ s.t.\ $N^+(w) = \{X_w\}$. 
  Then, Let $R$ be the binary relation on $W$ s.t.\ $v R u$ iff $\{u\} \in N^-(v)$.
  Note that $R$ is transitive: 
  if $v R u$ and $u R s$,
  i.e.\ $\{u\} \in N^-(v)$ and $\{s\} \in N^-(u)$,
  then by (NI$'''$)$^-$ and (Ser)$^-$,
  $\{s\} \in N^-(v)$, i.e.\ $v R s$.
  Then, consider $R(w)$.
  We show that $w \in R(w)$,
  and for any $v \in R(w)$ we have $v R v$
  and $N^+(v) = \{X_w\}$.
  By (Emp)$^+$, for any $u \in X_w$, $N^+(u) = \emptyset$.
  Hence, by our assumption that $N^+(w) \neq \emptyset$,
  $w \notin X_w$.
  Thus, by (Id)$^-$ and (Ser)$^-$, $\{w\} \in N^-(w)$,
  i.e.\ $w R w$.
  Next, take an arbitrary $v \in R(w)$.
  Then, $\{v\} \in N^-(w)$,
  so by (PI$''$)$^-$ and (Ser)$^-$, $\{v\} \in N^-(v)$,
  i.e.\ $v R v$.
  Moreover, since $\N$ is a core model and $\{v\} \in N^-(w)$,
  clearly $\{v\} \cap X_w = \emptyset$,
  so by (PI$''$)$^-$,
  $X_w \in {\uparrow}N^+(v)$,
  i.e.\ there is $X \in N^+(v)$ s.t.\ $X \subseteq X_w$.
  Then, by (NI$'''$)$^-$,
  $X \in N^+(w)$, so $X = X_w$ since $\N$ is a core model, and thus $X_w \in N^+(v)$.
  Thus, by ($1$-d-Bd)$^+$, $N^+(v) = \{X_w\}$.
  Finally, note that when $\N$ further satisfies (PI$'''$)$^-$,
  $R$ becomes Euclidean:
  if $v R u$ and $v R s$,
  i.e.\ $\{u\} \in N^-(v)$ and $\{s\} \in N^-(v)$,
  then by (PI$'''$)$^-$ and (Ser)$^-$,
  $\{s\} \in N^-(u)$, i.e.\ $u R s$.
  Then, since $w \in R(w)$,
  it follows that $R \cap (R(w))^2 = (R(w))^2$.

  Then, we define the Kripke model we need as $\M = (W^*, R^*, V^*)$,
  where $W^* = R(w) \cup X_w$,
  $R^* = (R \cap (R(w))^2) \cup (W^* \times X_w)$,
  and $V^*(p) = V(p) \cap W^*$ for all $p \in \BP$.
  By (Ser)$^+$, $X_w \neq \emptyset$, so $\M$ according to what we proved above, $\M$ is indeed an $S4.2$-model;
  moreover, when $\N$ satisfies (PI$'''$)$^-$,
  $R \cap (R(w))^2 = (R(w))^2 = (W^* \setminus X_w)^2$,
  so $\M$ is an $S4F$-model.
  Then, we could check that the identity function (restricted to $W^*$) satisfies the conditions listed in Definition \ref{def.rep} for $\tau_\mathtt{MB}$,
  so $\N,w$ is $\tau_\mathtt{MB}$-represented by $\M,w$.
\end{proof}

\begin{theorem}
  Under $\tau_\mathtt{MB}$-bundled semantics,
  $\mathbf{MBS4.2}$ is sound and strongly complete w.r.t.\ the class of $S4.2$-models,
  and $\mathbf{MBS4F}$ is sound and strongly complete w.r.t.\ the class of $S4F$-models.\footnote{
  Strictly speaking, the Kripke \emph{soundness} of $\mathbf{MBS4.2}$ and $\mathbf{MBS4F}$ under $\tau_\mathtt{MB}$-bundled semantics does not follow directly from their neighborhood soundness w.r.t.\ models with the corresponding properties listed in Table \ref{table.ax1} and Table \ref{table.ax2},
  since properties of the $\tau_\mathtt{MB}$-generated neighborhoods of $S4.2$- and $S4F$-Kripke models might be weaker than the ones listed in the tables.
  Nevertheless, they are strong enough to validate all corresponding axioms,
  which is rather routine to check,
  so we leave the details to the reader.}
\end{theorem}

\section{Conclusion}\label{sec.conclusion}

In this paper, we have studied propositional bundled modalities at a general level.
We showed how to view bundled semantics as a kind of neighborhood semantics
and uniformly defined bisimulations for bundled modalities from this viewpoint.
We also studied the expressivity and axiomatizations of PN-independent bundles
and offered axiomatizations for several concrete PN-independent bundles using our general method.

There are several directions for future work.
One is to further study the bundled bisimulation we defined.
For example, a \emph{characterization theorem} for these bisimulations is highly desirable.
A source of difficulty here is that
notions of modal saturation for ordinary neighborhood models appear too strong for the neighborhood structures generated by bundle terms,
so the usual ultrapower construction might not work here.

It also seems interesting to further study the class of PN-independent bundles.
For example, it might be interesting to try to offer a purely syntactic characterizion of the class of PN-independent bundles over a given class of Kripke models
--- though a grammar to generate PN-independent bundles is given in Proposition \ref{prop.cvnf},
it is still unclear whether the bundles it generates exhaust all PN-independent bundles (up to logical equivalence) over, say, the class of all models, or the class of serial models.

It should also be interesting to try to find more uniform ways to generate axiomatizations for bundled modalities.
For example, if we could prove Sahlqvist-style correspondence and completeness theorems for convex neighborhood logics,
and at the same time develop a method to automatically obtain representation results for certain bundles,
then we could immediately obtain a great amount of complete axiomatizations for these bundles.

Finally, we believe the class of convex neighborhood logics deserves a more fully developed theory.
On the one hand, some fine distinctions can only be made in this class of logics.
For example, the axioms $\mathtt{4}'_\bund$ and $\mathtt{4}''_\bund$ we introduced above both collapse to $\bund \phi \to \bund (\phi \wedge \bund \phi)$ in monotonic logics with $\mathtt{N}_\bund$,
and $\mathtt{4}'_\bund$, $\mathtt{4}''_\bund$, $\mathtt{4}'''_\bund$ all collapse to $\bund \phi \to \bund \bund \phi$ in normal modal logics.
On the other hand, logics in this class admit adequate semantics based on core models, so we may hope for a more fruitful theory than that of all neighborhood modal logics.
Hence, it should be interesting to further study the general theory of convex modal logics.

\appendix

\section{First-Order Undefinability of Bundled Bisimulation}\label{appd.fodf}

In this section we consider only bisimulation for $\tau = \Diamond\Box\Diamond(+)$ and show that it is not first-order definable in the sense that
there is no first-order formula $\Phi$ with predicates for $R$ and $V(p)$ ($p \in \mathbf{P}$) in Kripke models and a special predicate $Z$ such that $\Phi$ is satisfied by a Kripke model iff the interpretation of $Z$ is a $\tau$-bisimulation on that model.

Let $\M$ be the Kripke model $(W, R, V)$ where (as usual $2 = \{0, 1\}$)
\begin{align*}
  W & = \{w, u, v\} \cup \{w_i \mid i \in \Nat\} \cup \{w_{i, j} \mid i \in \Nat, j \in \{0, 1\}\} \cup 2^\Nat. \\
  R & = \{(w, u)\} \cup \{(u, w_i) \mid i \in \Nat\} \cup \{(w_i, w_{i, j}), (w_{i, j}, w_{i, j}) \mid i \in \Nat, j \in \{0, 1\}\} \cup \\[-0.3em]
  & \qquad (\{v\} \times 2^\Nat) \cup \{(f, w_{i, f(i)}) \mid f \in 2^\Nat, i \in \Nat\}.\\
  V & = \BP \times \{\emptyset\}.
\end{align*}
Note that $\Nbh_\M(w, \tau) = \Nbh_\M(v, \tau) = \{(\{w_{i, f(i)} \mid i \in \Nat\}, \emptyset) \mid f \in 2^\Nat\}$ and $\Dom_\M(w, \tau) = \Dom_\M(v, \tau) = \{w_{i, j} \mid i \in \Nat, j \in \{0, 1\}\}$.
Then, it is not hard to check that the binary relation $\Z$ defined as follows is a $\tau$-bisimulation on $\M$:
\begin{align*}
  \Z &= \{(w, v)\} \cup \{(w_{i, j}, w_{i, j}) \mid i \in \Nat, j \in \{0, 1\}\}.
\end{align*}

\begin{figure}[t]
  \centering
  \begin{tikzpicture}[
    >=Stealth,
    node distance=9mm and 10mm,
    every node/.style={inner sep=1.2pt}
  ]
    \begin{scope}
      \node (w) at (-5, -1) {$w$};
      \node (u) at (-4, -1) {$u$};
      \node (v) at (5.5, -1) {$v$};
      
      \node (w00) at (0, 0) {$w_{0,0}$};
      \node (w01) at (0, -1) {$w_{0,1}$};
      \node (w10) at (0, -2) {$w_{1,0}$};
      \node (w11) at (0, -3) {$w_{1,1}$};
      \node (w20) at (0, -4) {$w_{2,0}$};
      \node (w21) at (0, -5) {$w_{2,1}$};
      \node (w30) at (0, -6) {$w_{3,0}$};
      \node (w31) at (0, -7) {$w_{3,1}$};
      \node at (0, -7.5) {$\vdots$};
      \node at (0, -8) {$\vdots$};
      \node at (0, -8.5) {$\vdots$};

      \node (w0) at (-2, -0.5) {$w_0$};
      \node (w1) at (-2, -2.5) {$w_1$};
      \node (w2) at (-2, -4.5) {$w_2$};
      \node (w3) at (-2, -6.5) {$w_3$};
      \node (wdots) at (-2, -8.5) {$\vdots$};
      
      \node (ubd) at (3, 0) {};
      \node (dbd) at (3, -8) {$2^\Nat$};
      \node (1/4) at (3, -2) {$\bullet$};
      \node (5/8) at (3, -5) {$\bullet$};
      \node (13/16) at (3, -6.5) {$\bullet$};

      \draw[->] (w) to (u);
      \draw[->] (u) to (w0);
      \draw[->] (u) to (w1);
      \draw[->] (u) to (w2);
      \draw[->] (u) to (w3);
      \draw[->] (u) to (wdots);

      \draw[|-|] (ubd) to (dbd);
      \draw[-, dashed] (v) to (3.1, -0.1);
      \draw[-, dashed] (v) to (3.1, -7.7);
      \draw[->] (v) to (1/4);
      \draw[->] (v) to (5/8);
      \draw[->] (v) to (13/16);

      \draw[->] (w0) to (w00);
      \draw[->] (w0) to (w01);
      \draw[->] (w1) to (w10);
      \draw[->] (w1) to (w11);
      \draw[->] (w2) to (w20);
      \draw[->] (w2) to (w21);
      \draw[->] (w3) to (w30);
      \draw[->] (w3) to (w31);

      \draw[->] (1/4) to (w00);
      \draw[->] (1/4) to (w11);
      \draw[->] (1/4) to (w20);
      \draw[->] (1/4) to (w30);
      \draw[->] (1/4) to (0.1, -8.1);
    
      \draw[->] (5/8) to (0.3,-1.2);
      \draw[->] (5/8) to (w10);
      \draw[->] (5/8) to (w21);
      \draw[->] (5/8) to (w30);
      \draw[->] (5/8) to (0.1, -8.3);

      \draw[->] (13/16) to (w01);
      \draw[->] (13/16) to (w11);
      \draw[->] (13/16) to (w20);
      \draw[->] (13/16) to (w31);
      \draw[->] (13/16) to (0.1, -8.5);

      \draw[->] (w00) edge[loop above] (w00);
      \draw[->] (w01) edge[loop above] (w01);
      \draw[->] (w10) edge[loop above] (w10);
      \draw[->] (w11) edge[loop above] (w11);
      \draw[->] (w20) edge[loop above] (w20);
      \draw[->] (w21) edge[loop above] (w21);
      \draw[->] (w30) edge[loop above] (w30);
      \draw[->] (w31) edge[loop above] (w31);
    \end{scope}
  \end{tikzpicture}
  \caption{A schematic picture of the Kripke frame $(W,R)$ of $\M$ (the valuation is empty), in which
  $2^\Nat$ is represented by a segment of the real line.}
  \label{fig:tau-diamond-box-diamond-model}
\end{figure}

Observe also that when the part of $\M$ restricted to $W \setminus 2^\Nat$ remains intact,
removing any $f \in 2^\Nat$ from the model would result in $\Z$ being no longer a $\tau$-bisimulation.

If $\tau$-bisimulation were first-order definable, then we could take the Skolem hull of $\M$ from $W \setminus 2^\Nat$, and $\Z$ would still be a $\tau$-bisimulation in this Skolem hull.
But such a Skolem hull is a countable submodel of $\M$ and thus must omit one (indeed most) $f \in 2^\Nat$,
so $\Z$ in this submodel cannot be a $\tau$-bisimulation.
Hence, bisimulation for $\tau = \Diamond\Box\Diamond(+)$ is not first-order definable.

\section{Non-PN-Independence of Rumsfeld Ignorance}\label{appd.ncvri}
Let $\tau_\mathtt{RI}$ stand for the bundle $\Diamond (+) \wedge \Diamond (-) \wedge \Diamond (\Box (+) \vee \Box(-))$ of Rumsfeld ignorance,
and let $\tau^\op_\mathtt{RI}$ stand for its dual $\Box (+) \vee \Box (-) \vee \Box (\Diamond (+) \wedge \Diamond (-))$.
In order to show that $\tau_\mathtt{RI}$ (resp.\ $\tau^\op_\mathtt{RI}$) is not PN-independent over a model,
it suffices to show that the axiom of convexity $\mathtt{CONV}: \bund (\phi \wedge \psi) \wedge \bund (\phi \vee \chi) \to \bund \phi$
can be falsified in the model in question
under $\tau_\mathtt{RI}$-bundled (resp.\ $\tau^\op_\mathtt{RI}$-bundled) semantics.

The countermodels are the following ones:
  let $\M = (W, R, V)$ and $\N = (W', R', V')$,
  where 
  \begin{itemize}
    \item $W = \{w_0, w_1, w_2, w_3\}$, $W' = \{v_0, v_1\}$;
    \item $R = \{w_0, w_1\}^2 \cup (W \times \{w_2, w_3\})$, $R' = \{v_0\}^2 \cup (W' \times \{v_1\})$;
    \item $V(p) = \{w_0, w_2\}$, $V(q) = \{w_0, w_3\}$,
    and $V'(p) = \{v_1\}$, $V'(q) = \emptyset$.
  \end{itemize}
  See Figure \ref{fig.ncv} for a diagram of the above models.
  
  Then, $\M,w_0 \not \vDash_{\tau_\mathtt{RI}} \bund (p \wedge q) \wedge \bund (p \vee q) \to \bund p$,
  and $\N,v_0 \not \vDash_{\tau^\op_\mathtt{RI}} \bund (p \wedge q) \wedge \bund (p \vee \neg q) \to \bund p$.

  \begin{figure}[t]
    \centering
      \begin{tikzpicture}[
    >=Stealth,
    node distance=9mm and 10mm,
    every node/.style={inner sep=1.2pt}
  ]
    \begin{scope}
      \node[draw, circle, minimum width = .7cm, fill = lightgray!75] (w0) at (0, 0) {$p\,q$};
      \node[draw, circle, minimum width = .7cm] (w1) at (1.5, 0) {$\quad$};
      \node[draw, circle, minimum width = .7cm] (w2) at (0, 1.5) {$p$};
      \node[draw, circle, minimum width = .7cm] (w3) at (1.5, 1.5) {$q$};
      \node at (0.75, -1) {$\M,w_0$};
      
      \node[draw, circle, minimum width = .7cm] (v1) at (5, 1.5) {$p$};
      \node[draw, circle, minimum width = .7cm, fill = lightgray!75] (v0) at (5, 0) {$\quad$};
      \node at (5, -1) {$\N,v_0$};

      \draw[->] (w0) to (w2);
      \draw[->] (w0) to (w3);
      \draw[->] (w1) to (w2);
      \draw[->] (w1) to (w3);
      \draw[<->] (w0) to (w1);
      \draw[<->] (w2) to (w3);
      \draw[->] (w0) edge[out = 180, in = 270, looseness = 4] (w0);
      \draw[->] (w1) edge[out = 270, in = 0, looseness = 4] (w1);
      \draw[->] (w2) edge[out = 90, in = 180, looseness = 4] (w2);
      \draw[->] (w3) edge[out = 0, in = 90, looseness = 4] (w3);

      \draw[->] (v0) to (v1);
      \draw[->] (v1) edge[out = 45, in = 135, looseness = 4] (v1);
      \draw[->] (v0) edge[out = 225, in = 315, looseness = 4] (v0);
    \end{scope}
    \end{tikzpicture}
    \caption{A diagram of the countermodels $\M, w_0$ and $\N, v_0$.}

    \label{fig.ncv}
  \end{figure}

\section{Proof for Proposition \ref{prop.cvbis}}

\begin{proof}\label{appd.cvbis}
  It suffices to show that for any model $\M$, any $\tau \in \TT_\BA$, any binary relation $\Z$ on $\M$ and any $(w, v) \in \Z$,
  the conditions ($\tau$-Zig) and ($\tau$-Zag) imply ($\tau$-Coh),
  and the converse holds when $\tau$ is PN-independent over $\M$.
  Without loss of generality, we may also assume that $\Z$ is an equivalence relation, since we can always use the bisimilarity relation, which is itself a bisimulation.

  First, assume that $(w, v) \in \Z$ satisfies ($\tau$-Zig) and ($\tau$-Zag).
  Also assume that $X \subseteq \Dom_\M(w, \tau)$ and $Y \subseteq \Dom_\M(v, \tau)$ are $\Z$-coherent, and $X \in \Nbh^\cp_\M(w, \tau)$.
  We show that $Y \in \Nbh^\cp_\M(v, \tau)$.
  By definition, there is $(X_0, X_1) \in \Nbh_\M(w, \tau)$ s.t.\ $(X_0, X_1) \sqsubseteq (X, \Dom_\M(w, \tau) \setminus X)$,
  so $\Z[X_0] \cap X_1 = \emptyset$ by the $\Z$-coherence of $X$ and $Y$ (in particular, $\Z[X] \cap \Dom_\M(w, \tau) \subseteq X$).
  Hence, by ($\tau$-Zig), there is $(Y_0, Y_1) \in \Nbh_\M(v, \tau)$ s.t.\ $Y_0 \subseteq \Z[X_0]$,
  so $Y_0 \subseteq \Z[X_0] \cap \Dom_\M(v, \tau) \subseteq \Z[X] \cap \Dom_\M(v, \tau) \subseteq Y$;
  also, $Y_1 \subseteq \Z[X_1]$,
  and thus $Y_1 \subseteq \Z^{-1}[X_1]$ since $\Z$ is symmetric, which implies that $Y_1 \cap Y = \emptyset$, since otherwise there would be $w' \in X_1$ s.t.\ $w' \in \Z[Y] \cap \Dom_\M(w, \tau) \subseteq X$,
  causing a contradiction.
  Therefore, we have $(Y_0, Y_1) \sqsubseteq (Y, \Dom_\M(v, \tau) \setminus Y)$, and thus $Y \in \Nbh^\cp_\M(v, \tau)$.
  The converse direction of ($\tau$-Coh) can be proved symmetrically using ($\tau$-Zag).

  Next, assume that $\tau$ is PN-independent over $\M$ and $(w, v) \in \Z$ satisfies ($\tau$-Coh).
  We only show that $(w, v)$ satisfies ($\tau$-Zig), since the case for ($\tau$-Zag) is symmetric.
  Let $(X_0, X_1) \in \Nbh_\M(w, \tau)$ be arbitrary, and assume that $\Z[X_0] \cap X_1 = \emptyset$.
  First, let $X = \Z[X_0] \cap \Dom_\M(w, \tau)$ and $Y = \Z[X_0] \cap \Dom_\M(v, \tau)$.
  It is not hard to see that $X$ and $Y$ are $\Z$-coherent and $X \in \Nbh^\cp_\M(w, \tau)$.
  Thus, we have $Y \in \Nbh^\cp_\M(v, \tau)$,
  so there is $(Y_0, Y_1) \in \Nbh_\N(v, \tau)$ s.t.\ $Y_0 \subseteq Y \subseteq \Z[X_0]$.
  Next, let $X' = \Dom_\M(w, \tau) \setminus \Z[X_1]$ and $Y' = \Dom_\M(v, \tau) \setminus \Z[X_1]$.
  Again, $X'$ and $Y'$ are $\Z$-coherent and $X' \in \Nbh^\cp_\M(w, \tau)$,
  so $Y' \in \Nbh^\cp_\M(v, \tau)$.
  Hence, there is also $(Y'_0, Y'_1) \in \Nbh_\M(v, \tau)$ s.t.\ $Y'_1 \subseteq \Dom_\M(v, \tau) \setminus Y' \subseteq \Z[X_1]$.
  Thus, $(Y_0, Y'_1) \sqsubseteq (\Z[X_0], \Z[X_1])$.
  Finally, note that $\Z[X_0] \cap \Z[X_1] = \emptyset$,
  so $Y_0 \cap Y'_1 = \emptyset$.
  Then, since $\tau$ is PN-independent over $\M$, $(Y_0, Y'_1) \in \Nbh_\M(v, \tau)$.
\end{proof}

\section{Proofs for Lemma \ref{lem.cano1} and Lemma \ref{lem.cano2}}\label{appd.cano}

The following observations are useful.
\begin{proposition}\label{prop.filt}
  For any maximally consistent set $\Delta$ and $F \in \PC(\Delta)$,
  \begin{itemize}
    \item[(i)] There is $\phi \in F$ s.t.\ $\bund \phi \in \Delta$,
    and for any $\phi, \psi \in F$, if $\bund \phi \in \Delta$, then $\bund (\phi \wedge \psi) \in \Delta$.\footnote{
    In other words, $\{\phi \in F \mid \bund \phi \in \Delta\}$ is not only dense, but also \emph{open dense} in $F$,
    i.e.\ it is downward closed in $F$.
    }
    \item[(ii)] For any $\chi$, if for all $\phi \in F$, there is $\psi \in F$ s.t.\ $\vdash \psi \to \phi$ and $\bund (\psi \wedge \chi) \in \Delta$, then $\chi \in F$. 
  \end{itemize}

  The above also hold after replacing `$+$' with `$-$' and $\bund$ with $\nbund$.
\end{proposition}

\begin{proof}
We only prove the case for $+$-coherent filters, since the case for $-$-coherent filters is symmetric.

First, assume that $F$ is a filter $+$-coherent with $\Delta$.
Since $\top \in F$, by $+$-coherence, there is $\phi \in F$ s.t.\ $\bund (\top \wedge \phi) \in \Delta$ and thus $\bund\phi \in \Delta$ by $\mathtt{RE}$.
Moreover, for any $\phi, \psi \in F$, we have $\phi \wedge \psi \in F$,
so by $+$-coherence, there is $\chi \in F$ s.t.\ $\bund (\phi \wedge \psi \wedge \chi) \in \Delta$.
Thus, if $\bund \phi \in \Delta$, by $\mathtt{CONV}$ and $\mathtt{RE}$, we also have $\bund (\phi \wedge \psi) \in \Delta$.

Next, assume that for any $\phi \in F$, there is $\psi \in F$ s.t.\ $\vdash \psi \to \phi$ and $\bund (\psi \wedge \chi) \in \Delta$.
Then, note that the filter $G$ generated by $F \cup \{\chi\}$ is also $+$-coherent with $\Delta$,
and $F \subseteq G$.
Then, by the maximality of $F$, we must have $F = G$, so $\chi \in F$.
\end{proof}

We now present the proofs for Lemma \ref{lem.cano1} and Lemma \ref{lem.cano2}.
Proofs for the canonicity of $\mathtt{N}_\bund$, $\mathtt{D}_\bund$, $\mathtt{T}_\bund$ are omitted since they are rather trivial.
Moreover, we only prove canonicity for $\bund$-axioms,
since the cases for their $\nbund$-counterparts are completely symmetric.
Also note that the order in which the following axioms are dealt is slightly different from the order they are introduced in Table \ref{table.ax1} and Table \ref{table.ax2}.\\

(1) 
  \begin{tabular}{|ll|ll|}
    \hline
    $\mathtt{EQU}$ & $\bund \phi \to \nbund \phi$ & (Sym) & $N^+_w = N^-_w$\\
    \hline
  \end{tabular}

\begin{proof}
  Just notice that when $\mathtt{EQU}$ is an axiom,
  we have $\vdash \bund \phi \leftrightarrow \nbund \phi$ for all $\phi \in \LL_\bund$,
  so for any filter $F$, $F$ is $+$-coherent with $\Delta$ iff it is $-$-coherent with $\Delta$.
  Hence, we obviously have $\PC(\Delta) = \NC(\Delta)$, and thus $N^\oplus_\Delta = N^\ominus_\Delta$.
\end{proof}

(2)
  \begin{tabular}{|ll|ll|}
    \hline
    $\mathtt{C}^n_\bund$ & $\bigwedge_{0 \leq i \leq n} \bund \phi_i \to \bigvee_{0 \leq i < j \leq n} \bund (\phi_i \wedge \phi_j)$ & 
    ($n$-d-Bd)$^+$ & $|N^+_w| \leq n$\\
    \hline
  \end{tabular}

\begin{proof}
  Suppose that there are $F_0, ..., F_n \in \PC(\Delta)$ which are pairwise distinct.
  For each $i \leq n$, we list formulas in $F_i$ as $(\phi_i^k)_{k \in \Nat}$;
  by (i) of Proposition \ref{prop.filt}, we may also assume that $\bund \phi_i^0 \in \Delta$.
  Then, let $\psi^m_i = \bigwedge_{k \leq m} \phi^k_i$;
  by (i) of Proposition \ref{prop.filt}, $\bund \psi^m_i \in \Delta$.
  Then, for any $m \in \Nat$, we have $\bigwedge_{0 \leq i \leq n} \bund \psi^m_i \in \Delta$,
  so by $\mathtt{C}^n_\bund$, there are $0 \leq i < j \leq n$ s.t.\ $\bund (\psi^m_i \wedge \psi^m_j) \in \Delta$.
  Thus, there must be some $0 \leq i < j \leq n$ s.t.\ for infinitely many $m \in \Nat$,
  $\bund (\psi^m_i \wedge \psi^m_j) \in \Delta$.
  Then, consider the filter $G$ generated by $F_i \cup F_j$.
  For any $\phi_i^k \in F_i$, $\phi_j^l \in F_j$,
  we can find $m > \max \{k, l\}$ s.t.\ $\bund (\psi^m_i \wedge \psi^m_j) \in \Delta$.
  Also note that $\vdash \psi^m_i \wedge \psi^m_j \to \phi_i^k \wedge \phi_j^l$ and $\psi^m_i \wedge \psi^m_j \in G$.
  Thus, $G$ is $+$-coherent with $\Delta$, so by the maximality of $F_i$ and $F_j$, we must have $F_i = F_j = G$, causing a contradiction.
\end{proof}

(3) 
  \begin{tabular}{|ll|ll|}
    \hline
    $\mathtt{DIV}_\bund^n$ & $\bund \bigvee_{i \leq n} \phi_i \to \bigvee_{i \leq n} \bund \bigvee_{j \neq i} \phi_j$ &
    ($n$-s-Bd)$^+$ & $\forall X \in N^+_w |X| \leq n$\\
    \hline
  \end{tabular}
  
\begin{proof}
  Let $F \in \PC(\Delta)$ be arbitrary, and suppose that there are pairwise distinct $\Theta_0, ..., \Theta_n \in \ext{F}$.
  Then, we can find $\theta_0, ..., \theta_n$ s.t.\ $\vdash \bigvee_{i \leq n} \theta_i$,
  and for all $i \neq j \leq n$, $\theta_i \in \Theta_i$ and $\theta_i \notin \Theta_j$.
  Now, enumerate formulas in $F$ as $(\phi_k)_{k \in \Nat}$;
  by (i) of Proposition \ref{prop.filt},
  we may also assume that $\bund \phi_0 \in \Delta$.
  Then, let $\psi_m = \bigwedge_{k \leq m} \phi_k$;
  by (i) of Proposition \ref{prop.filt}, $\bund \psi_m \in \Delta$.
  Then, by $\mathtt{RE}$ and $\mathtt{DIV}^n_\bund$, for each $m \in \Nat$,
  there is $i \leq n$ s.t.\ $\bund (\psi_m \wedge \bigvee_{j \neq i} \theta_j) \in \Delta$.
  Hence, there is some $i \leq n$ for which there are infinitely many $m \in \Nat$
  s.t.\ $\bund (\psi_m \wedge \bigvee_{j \neq i} \theta_j) \in \Delta$.
  Then, for any $\phi \in F$, there is $m$ s.t.\ $\vdash \psi_m \to \phi$
  and $\bund (\psi_m \wedge \bigvee_{j \neq i} \theta_j) \in \Delta$,
  so by (ii) of Proposition \ref{prop.filt}, $\bigvee_{j \neq i} \theta_j \in F$.
  But then, since $F \subseteq \Theta_i$, $\bigvee_{j \neq i} \theta_j \in \Theta_i$, causing a contradiction.
\end{proof}

(4) 
  \begin{tabular}{|ll|}
    \hline
    $\mathtt{4}'_\bund$ & $\bund \phi \to \bund (\phi \wedge (\bund (\phi \vee \psi) \to \bund \phi))$ \\
    \hline
    (PI$'$)$^+$ & $\forall X \in N^+_w \forall v \in X (\exists Y \in N^-_v (X \cap Y = \emptyset) \Rightarrow X \in {\uparrow} N^+_v)$\\
    \hline
  \end{tabular}

\begin{proof}
  Let $F \in \PC(\Delta)$ and $\Theta \in \ext{F}$ be arbitrary.
  Also assume that there is $F' \in \NC(\Theta)$ s.t.\ $\ext{F} \cap \ext{F'} = \emptyset$.
  Our goal is to find $G \in \PC(\Theta)$ extending $F$, 
  so that $\ext{G} \subseteq \ext{F}$, and by Lemma \ref{lem.ext}, it is enough to show that $F$ is $+$-coherent with $\Theta$.
  Since $\ext{F} \cap \ext{F'} = \emptyset$, there is a $\phi \in F$ s.t.\ $\neg \phi \in F'$.
  Since $F'$ is $-$-coherent with $\Theta$, we may also assume that $\nbund \neg \phi \in \Theta$,
  i.e.\ $\bund \phi \in \Theta$.
  Now, to show that $F$ is $+$-coherent with $\Theta$, it is enough to show that for every $\psi \in F$,
  $\bund \phi \to \bund (\phi \wedge \psi) \in F$ since this implies that $\bund(\phi \wedge \psi) \in \Theta$.
  To do this, we use (ii) of Proposition \ref{prop.filt}.
  Let $\chi \in F$ be arbitrary.
  By $+$-coherence, there is $\theta \in F$ s.t.\ $\vdash \theta \to \phi \wedge \psi \wedge \chi$ and $\bund \theta \in \Delta$ (note $\phi, \psi, \chi \in F$).
  Thus, by $\mathtt{4}'_\bund$ with $\theta$ for $\phi$ and $\phi$ for $\psi$ and detaching $\bund\theta \in \Delta$ and replacing $\theta \lor \phi$ by the equivalent $\phi$, $\bund (\theta \wedge (\bund \phi \to \bund \theta)) \in \Delta$.
  Note that, by $\mathtt{CONV}$, $\vdash (\bund \phi \to \bund \theta) \to (\bund \phi \to \bund (\phi \wedge \psi))$,
  so by $\mathtt{CONV}$ again (easier to see in its rule form), we have 
  $\vdash \bund \theta \wedge \bund (\theta \wedge (\bund \phi \to \bund \theta)) \to \bund (\theta \wedge (\bund \phi \to \bund (\phi \wedge \psi)))$.
  Summarizing above, we have both $\vdash \theta \to \chi$ and $\bund (\theta \wedge (\bund \phi \to \bund (\phi \wedge \psi))) \in \Delta$,
  completing the requirement for (ii) of Proposition \ref{prop.filt}. So $\bund \phi \to \bund (\phi \wedge \psi) \in F$.
\end{proof}

(5) 
  \begin{tabular}{|ll|}
    \hline
    $\mathtt{4}''_\bund$ & $\bund \phi \wedge \bund (\phi \vee \psi) \to \bund (\phi \wedge \bund (\phi \vee \psi))$ \\
    \hline
    (PI$''$)$^+$ & $\forall X \in N^+_w \forall v \in X (X \in {\uparrow} N^+_v \wedge \forall Y \in N^-_w (X \cap Y = \emptyset \Rightarrow Y \in {\uparrow} N^-_v))$\\
    \hline
  \end{tabular}

\begin{proof}
  Let $F \in \PC(\Delta)$ and $\Theta \in \ext{F}$ be arbitrary.
  First, for any $\phi, \psi \in F$,
  if $\bund \phi \in \Delta$,
  then $\bund (\phi \wedge \psi) \in \Delta$ by (i) of Proposition \ref{prop.filt},
  and thus, by $\mathtt{4}''_\bund$ ($\phi \land \psi$ for $\phi$) and $\mathtt{RE}$, $\bund (\phi \wedge \psi \wedge \bund \phi) \in \Delta$.
  Hence, by (ii) of Proposition \ref{prop.filt}, $\bund \phi \in F \subseteq \Theta$.
  Then, it is not hard to see that $F$ is also $+$-coherent with $\Theta$,
  so by Lemma \ref{lem.ext}, there is $G \in \PC(\Theta)$ extending $F$, and thus $\ext{G} \subseteq \ext{F}$.

  Next, let $F' \in \NC(\Delta)$ be arbitrary, and assume that $\ext{F} \cap \ext{F'} = \emptyset$.
  Then, there is $\phi \in F$ s.t.\ $\neg \phi \in F'$,
  and we may assume that $\nbund \neg \phi \in \Delta$ (i.e.\ $\bund \phi \in \Delta$) by the $-$-coherence of $F'$.
  Then, for any $\psi \in F'$, 
  we have $\nbund (\neg \phi \wedge \psi) \in \Delta$ by (i) of Proposition \ref{prop.filt},
  i.e.\ $\bund (\phi \vee \neg \psi) \in \Delta$.
  Then, for any $\chi \in F$, since $\bund (\phi \wedge \chi) \in \Delta$,
  by $\mathtt{4}''_\bund$,
  $\bund (\phi \wedge \chi \wedge \nbund (\neg \phi \wedge \psi)) \in \Delta$.
  Thus, by (ii) of Proposition \ref{prop.filt},
  $\nbund (\neg \phi \wedge \psi) \in F \subseteq \Theta$.
  In sum, $\psi \in F' \Rightarrow \nbund(\lnot\phi \land \psi) \in \Theta$, so $F'$ is $-$-coherent with $\Theta$,
  so there is $G' \in \NC(\Theta)$ extending $F'$, and thus $\ext{G'} \subseteq \ext{F'}$.
\end{proof}

(6) 
  \begin{tabular}{|ll|ll|}
    \hline
    $\mathtt{4}'''_\bund$ & $\bund \phi \wedge \bund \psi \to \bund (\phi \wedge \bund \psi)$ &
    (PI$'''$)$^+$ & $\forall X \in N^+_w \forall v \in X (N^+_w \subseteq {\uparrow}N^+_v \wedge N^-_w \subseteq {\uparrow}N^-_v)$\\
    \hline
  \end{tabular}

\begin{proof}
  Let $F \in \PC(\Delta)$ and $\Theta \in \ext{F}$ be arbitrary.
  By $\mathtt{4}'''_\bund$ and (ii) of Proposition \ref{prop.filt},
  for any $\bund \psi \in \Delta$,
  we have $\bund \psi \in F \subseteq \Theta$.
  Thus, if a filter $G$ is $+$-coherent (resp.\ $-$-coherent) with $\Delta$,
  it would also be $+$-coherent (resp.\ $-$-coherent) with $\Theta$,
  so by Lemma \ref{lem.ext}, there is $H \in \PC(\Theta)$ (resp.\ $\NC(\Theta)$) extending $G$,
  and thus $\ext{H} \subseteq \ext{G}$.
\end{proof}

(7) 
  \begin{tabular}{|ll|ll|}
    \hline
    $\mathtt{5}'''_\bund$ & $\bund \phi \wedge \neg \bund \psi \to \bund (\phi \wedge \neg \bund \psi)$ &
    (NI$'''$)$^+$ & $\forall X \in N^+_w \forall v \in X (N^+_v \subseteq {\uparrow}N^+_w \wedge N^-_v \subseteq {\uparrow}N^-_w)$\\
    \hline
  \end{tabular}

\begin{proof}
  Let $F \in \PC(\Delta)$ and $\Theta \in \ext{F}$ be arbitrary.
  By $\mathtt{5}'''_\bund$ and (ii) of Proposition \ref{prop.filt},
  for any $\bund \psi \notin \Delta$,
  we have $\neg \bund \psi \in F \subseteq \Theta$, and thus $\bund \psi \notin \Theta$.
  Thus, if $G$ is $+$-coherent (resp.\ $-$-coherent) with $\Theta$,
  it would also be $+$-coherent (resp.\ $-$-coherent) with $\Delta$,
  so by Lemma \ref{lem.ext}, there is $H \in \PC(\Delta)$ (resp.\ $\NC(\Delta)$) extending $G$, and thus $\ext{H} \subseteq \ext{G}$.
\end{proof}

(8) 
  \begin{tabular}{|ll|}
    \hline
    $\mathtt{5C}^n_\bund$ & $\bund \phi \wedge \bigwedge_{0 \leq i \leq n} \neg \bund (\phi \wedge \psi_i) \to \bund (\phi \wedge (\bigwedge_{0 \leq i \leq n} \bund \psi_i \to \bigvee_{0 \leq i < j \leq n} \bund (\psi_i \wedge \psi_j)))$\\
    \hline
    ($n$-Coa) & $\forall X \in N^+_w \forall v \in X |N^+_v \setminus {\uparrow}\{X\}| \leq n$\\
    \hline
  \end{tabular}

\begin{proof}
  Let $F \in \PC(\Delta)$, $\Theta \in \ext{F}$ be arbitrary,
  and suppose that there are pairwise distinct $G_0, ..., G_n \in \PC(\Theta)$ s.t.\ for each $i \leq n$, $\ext{F} \not \subseteq \ext{G_i}$.
  Then, for each $i \leq n$, there is $\phi_i^0 \in G_i \setminus F$;
  since $G_i$ is $+$-coherent with $\Theta$, we may also assume that $\bund \phi_i^0 \in \Theta$.
  Now, we list formulas in $G_i$ as $(\phi_i^k)_{k \in \Nat}$,
  where $\phi_i^0$ is selected in the way above;
  then, for each $m \in \Nat$, let $\psi^m_i = \bigwedge_{k \leq m} \phi_i^k$.
  By (i) of Proposition \ref{prop.filt}, $\bund \psi^m_i \in \Theta$.
  Then, let $m \in \Nat$ be arbitrary;
  we show that $\bigwedge_{0 \leq i \leq n} \bund \psi^m_i \to \bigvee_{0 \leq i < j \leq n} \bund (\psi^m_i \wedge \psi^m_j) \in \Theta$.
  Since $\phi_i^0 \notin F$, clearly $\psi^m_i \notin F$,
  so by (ii) of Proposition \ref{prop.filt} in contraposition, there must be $\chi^m_i \in F$ s.t.\
  $\bund \chi^m_i \in \Delta$, but
  for any $\theta \in F$, $\bund (\chi^m_i \wedge \theta \wedge \psi^m_i) \notin \Delta$.
  Then, let $\chi^m = \bigwedge_{i \leq n} \chi_i^m$.
  For any $\theta \in F$ and $i \leq n$,
  $\bund (\chi^m \wedge \theta) \in \Delta$ (by (i) of Proposition \ref{prop.filt}) and  
  $\bund (\chi^m \wedge \theta \wedge \psi^m_i) \notin \Delta$,
  so by $\mathtt{5C}^n_\bund$, $\bund (\chi^m \wedge \theta \wedge (\bigwedge_{0 \leq i \leq n} \bund \psi^m_i \to \bigvee_{0 \leq i < j \leq n} \bund (\psi^m_i \wedge \psi^m_j))) \in \Delta$.
  Then, since $\theta \in F$ is arbitrary, by (ii) of Proposition \ref{prop.filt}, $\bigwedge_{0 \leq i \leq n} \bund \psi^m_i \to \bigvee_{0 \leq i < j \leq n} \bund (\psi^m_i \wedge \psi^m_j) \in F \subseteq \Theta$.
  The rest is the same as the case for $\mathtt{C}^n_\bund$.
\end{proof}

(9)
  \begin{tabular}{|ll|ll|}
    \hline
    $\mathtt{DE}_\bund$ & $\bund \phi \to \bund (\phi \wedge \neg \bund \psi)$ &
    (Emp)$^+$ & $\forall X \in N^+_w \forall v \in X (N^+_v = N^-_v = \emptyset)$\\
    \hline
  \end{tabular}

\begin{proof}
  Let $F \in \PC(\Delta)$ and $\Theta \in \ext{F}$ be arbitrary.
  By $\mathtt{DE}_\bund$, for any $\psi$ and any $\phi \in F$ s.t.\ $\bund \phi \in \Delta$,
  we have $\bund (\phi \wedge \neg \bund \psi) \in \Delta$.
  Thus, by (ii) of Proposition \ref{prop.filt}, $\neg \bund \psi \in F \subseteq \Theta$.
  Therefore, there is no $\psi$ s.t.\ $\bund \psi \in \Theta$,
  so $\PC(\Theta) = \NC(\Theta) = \emptyset$.
\end{proof}

(10)
  \begin{tabular}{|ll|ll|}
    \hline
    $\mathtt{T}^\mathtt{op}_\bund$ & $\phi \wedge \bund (\phi \vee \psi) \to \bund \phi$ &
    (Id)$^+$ & $(\exists Y \in N^-_w w \notin Y) \Rightarrow \{w\} \in {\uparrow} N^+_w$\\
    \hline
  \end{tabular}

\begin{proof}
  Assume that there is $F' \in \NC(\Delta)$ s.t.\ $\Delta \notin \ext{F'}$.
  Then, there is $\phi \in F'$ s.t.\ $\neg \phi \in \Delta$;
  since $F'$ is $-$-coherent w.r.t.\ $\Delta$,
  we may also assume that $\bund \neg \phi \in \Delta$.
  Then, let $\psi \in \Delta$ be arbitrary.
  Since $\neg \phi \wedge \psi \in \Delta$ and $\bund \neg \phi \in \Delta$,
  by $\mathtt{T}^\mathtt{op}_\bund$, $\bund (\neg \phi \wedge \psi) \in \Delta$.
  Thus, $\Delta$, viewed as a filter, is $+$-coherent with $\Delta$ itself.
  Then, by Lemma \ref{lem.ext}, there is $G \in \PC(\Delta)$ extending $\Delta$, and thus $\ext{G} \subseteq \{\Delta\}$.
\end{proof}

\bibliographystyle{eptcs}
\bibliography{generic}

\end{document}